\PassOptionsToPackage{compress}{natbib}
\documentclass{article}
\usepackage{ijcai26}
\usepackage[T1]{fontenc}
\newcommand{\citet}[1]{\citeauthor{#1}~[\citeyear{#1}]}
\newcommand{\citep}[1]{\cite{#1}}
\usepackage{soul}
\usepackage{url}
\usepackage[hidelinks]{hyperref}
\usepackage[utf8]{inputenc}
\usepackage[small]{caption}
\usepackage{graphicx}
\usepackage{amsmath}
\usepackage{amsthm}
\usepackage{booktabs}
\usepackage[switch]{lineno}
\usepackage{times}
\usepackage{thmtools,thm-restate}
\usepackage{tabularx}
\usepackage{tikz}
\usetikzlibrary{calc}
\usepackage{graphicx}
\usepackage{caption}
\usepackage{subcaption}
\usepackage{amsmath,amsthm,mathtools}
\usepackage{amssymb}
\usepackage{comment}
\usepackage{tikz}
\usetikzlibrary{decorations}
\usetikzlibrary{decorations.pathreplacing}

\usepackage{algorithm}

\usepackage{newfloat}
\usepackage{listings}

\usepackage{enumitem}

\usepackage{pgfplots}
\pgfplotsset{compat=1.18}
\usepackage{pgfplotstable}

\usepackage{multirow}
\usepackage{amsmath, amsthm, amssymb, mathtools}
\usepackage{amsfonts}
\usepackage{thmtools}
\usepackage{dsfont}
\usepackage{tikz} 
\usepackage{makecell}
\usepackage[sort&compress,capitalise,nameinlink]{cleveref}
\usepackage{nicefrac}
\usepackage{comment}
\usepackage{tabularx}
\usepackage{csquotes}
\usepackage{thm-restate}
\usepackage[noend]{algpseudocode}
\usepackage{booktabs}
\newtheorem{example}{Example}
\newtheorem{theorem}{Theorem}
\renewcommand*{\le}{\leqslant}
\renewcommand*{\leq}{\leqslant}
\renewcommand*{\ge}{\geqslant}
\renewcommand*{\geq}{\geqslant}
\renewcommand{\epsilon}{\varepsilon}

\usepackage{resizegather}
\newcommand{\N}{V}
\newcommand{\x}{\boldsymbol{x}}
\newcommand{\Btau}{\boldsymbol{\tau}}
\newcommand{\p}{\boldsymbol{p}}
\renewcommand{\d}{\boldsymbol{d}}
\newcommand{\de}{\boldsymbol{d}}
\newcommand{\inst}{\mathcal{I}}

\newcommand{\ODP}{\texttt{ODP}}

\crefname{table}{Table}{Tables}
\crefname{figure}{Figure}{Figures}
\crefname{theorem}{Theorem}{Theorems}
\crefname{definition}{Definition}{Definitions}
\crefname{corollary}{Corollary}{Corollaries}
\crefname{observation}{Observation}{Observations}
\crefname{question}{Question}{Question}
\crefname{lemma}{Lemma}{Lemmas}
\crefname{claim}{Claim}{Claims}
\crefname{example}{Example}{Examples}
\crefname{reduction}{Reduction}{Reductions}
\crefname{construction}{Construction}{Constructions}
\crefname{subsection}{Section}{Sections}
\crefname{section}{Section}{Sections}
\crefname{proposition}{Proposition}{Propositions}
\crefname{algorithm}{Algorithm}{Algorithms}
\Crefname{equation}{Expression}{Expressions}
\crefname{lstlisting}{listing}{listings}
\crefname{appendix}{Appendix}{Appendices}

\renewcommand{\part}{{{\mathrm{part}}}}

\newcommand{\dist}{{{\mathrm{dist}}}}

\renewcommand{\phi}{\varphi}
\renewcommand{\leq}{\leqslant}
\renewcommand{\geq}{\geqslant}
\DeclareMathOperator*{\argmin}{argmin}
\DeclareMathOperator*{\argmax}{argmax}

\usepackage{listofitems}
\definecolor{darkred}{rgb}{0.64,0,0}
\definecolor{darkcyan}{rgb}{0,0.55,0.55}
\newcommand{\rowcolor}[1]{\textcolor{black}{#1}}
\newcommand{\columncolor}[1]{\textcolor{black}{#1}}

\newcommand{\nfgameN}[1]{%
\setsepchar{ }
\readlist\arg{#1}
\begin{tikzpicture}[scale=1.8]
	\node (RT) at (-2.2,1.1) [label=left:\rowcolor{ \arg[1] }] {};
\node (RB) at (-2.2,-1.1) [label=left:\rowcolor{\arg[2]}] {};
\node (CL) at (-1.1,2.2) [label=above:\columncolor{\arg[3]}] {};
\node (CR) at (1.1,2.2) [label=above:\columncolor{\arg[4]}] {};
\node (RTL) at (-1.1,1.1) {\rowcolor{\hspace{0.1cm}\arg[5]}}; 
\node (RBL) at (-1.1,-1.1) {\rowcolor{\hspace{0.1cm}\arg[6]}};
\node (RTR) at (1.1,1.1) {\rowcolor{\hspace{0.1cm}\arg[7]}};
\node (RBR) at (1.1,-1.1) {\rowcolor{\hspace{0.1cm}\arg[8]}};
\draw[-,very thick] (-2.2,-2.2) to (2.2,-2.2);
\draw[-,very thick] (-2.2,0) to (2.2,0);
\draw[-,very thick] (-2.2,2.2) to (2.2,2.2);
\draw[-,very thick] (-2.2,-2.2) to (-2.2,2.2);
\draw[-,very thick] (0,-2.2) to (0,2.2);
\draw[-,very thick] (2.2,-2.2) to (2.2,2.2);
\end{tikzpicture}}

\newcommand{\nfgametwothree}[1]{%
\setsepchar{ }
\readlist\arg{#1}
\begin{tikzpicture}[scale=0.9]
\node (RT) at (-3,1)  [label=left:\rowcolor{\arg[1]}] {};
\node (RB) at (-3,-1) [label=left:\rowcolor{\arg[2]}] {};
\node (CL) at (-2,2)  [label=above:\columncolor{\arg[3]}] {};
\node (CM) at (0,2)  [label=above:\columncolor{\arg[4]}] {};
\node (CR) at (2,2)  [label=above:\columncolor{\arg[5]}] {};

\node (RTL) at (-2.4,0.6) {\rowcolor{\arg[6]}}; 
\node (CTL) at (-1.6,1.4) {\columncolor{\arg[7]}};
\node (RBL) at (-2.4,-1.4) {\rowcolor{\arg[12]}};
\node (CBL) at (-1.6,-0.6) {\columncolor{\arg[13]}};

\node (RTM) at (-0.4,0.6) {\rowcolor{\arg[8]}}; 
\node (CTM) at (0.4,1.4) {\columncolor{\arg[9]}};
\node (RBM) at (-0.4,-1.4) {\rowcolor{\arg[14]}};
\node (CBM) at (0.4,-0.6) {\columncolor{\arg[15]}};

\node (RTR) at (1.6,0.6) {\rowcolor{\arg[10]}}; 
\node (CTR) at (2.4,1.4) {\columncolor{\arg[11]}};
\node (RBR) at (1.6,-1.4) {\rowcolor{\arg[16]}};
\node (CBR) at (2.4,-0.6) {\columncolor{\arg[17]}};

\draw[-,very thick] (-3,-2) to (3,-2);
\draw[-,very thick] (-3,0) to (3,0);
\draw[-,very thick] (-3,2) to (3,2);
\draw[-,very thick] (-3,-2) to (-3,2);
\draw[-,very thick] (-1,-2) to (-1,2);
\draw[-,very thick] (1,-2) to (1,2);
\draw[-,very thick] (3,-2) to (3,2);
\draw[-,very thin] (-3,2) to (1,-2);
\draw[-,very thin] (-1,2) to (3,-2);
\draw[-,very thin] (-3,0) to (-1,-2);
\draw[-,very thin] (1,2) to (3,0);
\draw[-,very thin] (-3,-2) to (-1,-2);
\end{tikzpicture}}

\usetikzlibrary{patterns} 

\usepackage{tablefootnote}
\theoremstyle{definition}
\newtheorem{definition}{Definition}

\theoremstyle{plain}
\newtheorem{lemma}[theorem]{Lemma}
\newtheorem{corollary}[theorem]{Corollary}

\newtheorem{proposition}[theorem]{Proposition}
\newtheorem{observation}[theorem]{Observation}
\newtheorem*{example1cont}{\Cref{Longterm:ex:start} Continued}

\newtheoremstyle{noparentheses}
  {\topsep}   
  {\topsep}   
  {\itshape}  
  {0pt}       
  {\bfseries} 
  {.}         
  {5pt plus 1pt minus 1pt} 
  {\thmname{#1}\thmnote{#3}}    
\theoremstyle{noparentheses}

\usepackage{soul}

\definecolor{dred}{rgb}{0.48,0.09,0.09}

\definecolor{brightube}{rgb}{0.82, 0.62, 0.91}

\newtheoremstyle{def}
  {}{}
  {\itshape}
  {}
  {\bfseries}
  {.}
  { }
  {%
   \thmname{#1}
   \thmnumber{ #2}
   \thmnote{ {\bfseries\ifdef(\fi#3\ifdef)\fi}}
  }

{\theoremstyle{def}{
}}

\newtheoremstyle{ex}
{}{}
{}
{}
{\bfseries}
{.}
{ }
{%
	\thmname{#1}
	\thmnumber{ #2}
	\thmnote{{\bfseries: \ifex(\fi#3\ifex) \fi}}
}
\newif\ifex

{\theoremstyle{ex}{
		}}
\Crefname{exof}{Example}{Examples}

\pgfplotsset{compat=1.11,
    /pgfplots/ybar legend/.style={
    /pgfplots/legend image code/.code={%
       \draw[##1,/tikz/.cd,yshift=-0.25em]
        (0cm,0cm) rectangle (3pt,0.8em);},
   },
}

\expandafter\def\expandafter\normalsize\expandafter{%
    \normalsize%
    \setlength\abovedisplayskip{-5pt}%
    \setlength\belowdisplayskip{2pt}%
    \setlength\abovedisplayshortskip{-8pt}%
    \setlength\belowdisplayshortskip{3pt}%
}
\title{Cycles in Liquid Democracy: A Game-Theoretic Justification} 
\author{
Markus Brill$^1$
\and
Rachael Colley$^2$\and
Anne-Marie George$^{3}$\and
Grzegorz Lisowski$^4$\and\\
Georgios Papasotiropoulos$^5$ \And
Ulrike Schmidt-Kraepelin$^6$\\
\affiliations
$^1$University of Oxford\\
$^2$University of Glasgow\\
$^3$University of Oslo\\
$^4$University of Groningen\\
$^5$University of Warsaw\\
$^6$TU Eindhoven\\
\emails
markus.brill@cs.ox.ac.uk,
rachael.colley@glasgow.ac.uk,
annemage@ifi.uio.no,
g.a.lisowski@rug.nl,
gpapasotiropoulos@gmail.com,
u.schmidt.kraepelin@tue.nl
}

\begin{document}

\maketitle


\begin{abstract}
A common criticism of liquid democracy within the relevant academic literature is that \textit{delegation cycles} can occur, seemingly resulting in unused voting power. Yet, practitioners argue that delegation cycles are not only unproblematic but are even formed intentionally by participants. 
To bring theory closer to reality, we introduce a model that captures this strategic behavior under uncertainty. We study the existence, structure and quality of Nash equilibria, revealing that delegation cycles naturally emerge. To complement these findings, we perform computational experiments using best-response dynamics.\looseness=-1
\end{abstract}

%

\section{Introduction}

Liquid democracy is a flexible voting system that allows voters to vote directly or \emph{delegate} their vote to another participant, who can vote on their behalf. Delegations are \emph{transitive}, meaning that delegated votes can be delegated further, creating \emph{delegation chains}. The voter at the end of a chain casts ballots on behalf of everyone in the chain.
This system combines the advantages of \emph{direct democracy} and \emph{representative democracy} by giving voters the freedom to choose their own mode of participation \citep{BlZu16x}.

Over the past decade, academic interest in liquid democracy has grown rapidly \citep{liquid_survey,papasotiropoulosliquid}. 
However, practitioners have noted that some parts of this literature overlook or misinterpret key aspects of liquid democracy as it is implemented in practice \citep{behrens2014principles,behrens2015circular}. While liquid democracy is often applied over an extended period of time to an ongoing stream of decisions, much of the literature models it as a one-time event. This divergence has led to differing views on the issue of \textit{delegation cycles}. In the literature, \textit{delegating} and \textit{voting} are frequently modeled as mutually exclusive options. Consequently, \textit{delegation cycles} (i.e., 
a situation where a voter~$i$ delegates to~$j$, and this delegation eventually returns to~$i$ through a chain of delegations) are seen as problematic because none of the voters in the cycle eventually casts a vote, leaving their collective voting weight unused. In contrast, \citet{BKNS21x} propose an alternative interpretation (informed by their experience with the \textit{LiquidFeedback} platform): voters specify \textit{default delegations} that remain fixed across multiple decisions. These default delegations serve as a fallback whenever voters do not cast a vote; if a voter casts a vote, their default delegation is ignored. 
Under this interpretation, like-minded voters who trust one another may intentionally form delegation cycles with their default delegations. These cycles rarely result in lost votes for actual decisions, since all the voting weight in the cycle is used as long as at least one of the involved voters casts a vote.

To bridge theory and practice, we introduce the \emph{default delegation model}. Voters declare delegations for a future election, which will be used only if they do not participate themselves. To model uncertainty about who will participate in this election, each voter has a fixed probability of participating in the election. A voter's \emph{utility} thus depends on both 
(i) other voters’ delegation choices and 
(ii) who actually turns out to vote. We assume that participants live in a (one-dimensional) metric space and prefer to be represented by those who are close to them. Moreover, voters can have different \emph{tolerance levels} towards being represented by far-away voters. While the model is not explicitly temporal, analyzing strategic behavior in this probabilistic setting can shed light on long-term delegation dynamics. In a world with a stream of similar elections, where each voter participates in only a fraction of them and participation is not highly correlated, equilibrium states in our model can serve as proxies for long-term behavior across the election stream.

We provide theoretical evidence supporting the practical observation that delegation cycles naturally arise among rational users of liquid democracy platforms. More concretely, 
within the default delegation model,
we provide a game-theoretic analysis of how voters can strategically set their default delegations to ensure that their voting power ends up with casting voters whose preferences closely align with their own.
We study (i) \emph{existence}, (ii) \emph{structure}, and (iii) \emph{quality} of (pure) Nash equilibria and demonstrate that, under mild assumptions, delegation cycles are necessarily formed. 
Our findings are complemented by computational experiments, providing additional insights into best-response dynamics.\looseness-1

\subsection{Related Work}
\label{sec:related}
Recent research in (computational) social choice and beyond shows a growing interest in liquid democracy, with various models and methodologies emerging. 
We focus on studies that are related to ours in terms of these two aspects.

\paragraph{The Role of Cyclic Delegations.}
A key branch of the literature addresses delegation cycles, typically viewing them as undesirable and proposing solutions to eliminate them \citep{brill2022liquidranked,cycleschristoff2016liquid,CGN22x,parameterized,cyclejain2022preserving,cyclekoppe2022fine,cycleskotsialou2018incentivising,MaPa25a,tyrovolas2024unravelling,utke2024anonymousu}. All of these works distinguish between delegating and casting voters and focus on axiomatic and algorithmic aspects, rather than strategic behavior. 
Notably, the work by \citet{MaPa25a}, like ours, was directly motivated by the study of \citet{BKNS21x}.
They study a temporal model with delegation updates over discrete time-steps. By contrast, in our paper, the temporal aspect is captured through the probabilistic model for ballot casting. 
Other approaches prevent cycles by design: 
\citet{FlexRep} 
disallow transitive delegations; 
\citet{kahng2021liquid} only allow delegations to voters with higher competence;
\citet{caragiannis2019contribution} assume the existence of a mechanism preventing cyclic delegations.
\vspace{-0.1cm}
\paragraph{One-Dimensional Spatial Models.}
To model voters' preferences over potential delegates, we assume that voters are positioned in a one-dimensional metric space. This is a common modeling choice (often representing ideological alignment), used in delegative voting settings by \citet{yamakawa2007toward}, \citet{Gree15x}, \citet{CMM+17an}, \citet{escoffier2020iterative}, and \citet{anshelevich2021representative}.

\vspace{-0.1cm}
\paragraph{Strategic Delegation Behavior.}
A prominent line of research in liquid democracy employs a game-theoretic perspective. Specifically, this includes work that analyzes Nash equilibria of delegation games \citep{NEbersetche2022generalizing,bloembergen2019rational,EGP19x,escoffier2020iterative}, provides worst-case guarantees \citep{noel2021pirates}, or studies voting power \citep{zhang2021power,CDG23}.
However, the games that these works consider are significantly different from the ones we analyze. None of them models the strategic choice of delegates under probabilistic voter participation. For instance, prior utility models are based on the effort of voting \citep{bloembergen2019rational} or are analyzed in purely deterministic settings \citep{EGP19x,escoffier2020iterative}.

\subsection{Our Contribution}  
Our central contribution is conceptual: we introduce the novel \emph{default delegation model} for liquid democracy. It captures and explains strategic delegation behavior under uncertain participation\,---\,a feature inherent in real-world liquid democracy systems.
\noindent To provide insights into the occurrence, merits, and soundness of cycles, and to establish the first principled justification for their emergence, we analyze the following:

\vspace{-0.1cm}
\paragraph{Existence of Nash Equilibria.} 
Our computational experiments suggest that Nash equilibria are prevalent across a broad range of synthetic instances.
In \Cref{sec:existence}, on the negative side, we identify instances where a Nash equilibrium does not exist, even in simple settings with only three voters, or where all voters have identical tolerance levels. On the positive side, we establish the existence of Nash equilibria in several special cases or slight variants of our original model.
\vspace{-0.1cm}
\paragraph{Structure of Nash Equilibria.} 
In \Cref{sec:structure}, we prove that, under mild assumptions, strategic voters form delegation cycles in equilibrium. More precisely, every non-trivial component of an equilibrium delegation graph contains exactly one cycle. Furthermore, we show that relaxing any of the assumptions invalidates the result, and we provide additional insights into the structure of delegations at equilibrium.
In more general settings, computational experiments reveal that the vast majority of components contain cycles. The width of these cycles appears to be proportional to voters' tolerance levels and inversely proportional to the number of voters.
\vspace{-0.1cm}
\paragraph{Quality of Nash Equilibria.} In \Cref{sec:qualityofNE} we evaluate the quality of Nash equilibria primarily in terms of their social welfare, i.e., the total utility they achieve, and we measure the \textit{Price of Anarchy (PoA)}, i.e., the ratio between the best possible social welfare and the social welfare of equilibria. 
While we prove that the PoA is generally unbounded, we also provide strong positive results: for non-degenerate instances, the difference between the two quantities is bounded, and as voting probabilities increase or tolerance levels decrease, the welfare in equilibrium approaches the optimal social welfare.
Moreover, our experiments show that Nash equilibria often achieve close-to-optimal social welfare.

\smallskip
\noindent All proofs and omitted details can be found in the Appendix.
To show how they complement the theory, our experimental findings are presented right after the results they support.

\section{The Default Delegation Model}
\label{sec:model}

We consider a finite set $\N$ of \emph{voters} using a liquid democracy platform. Each voter nominates a \emph{default delegate} for a future election, in which any voter may choose to vote or abstain.  
If a voter $i \in \N$ does \emph{not} vote, their voting power is passed to their default delegate, continuing transitively until a voter who casts a vote is reached\,---\,this voter is called $i$’s \emph{ultimate delegate}.  
If no one in the delegation chain votes, $i$ has no ultimate delegate and their voting power is lost.

\vspace{-0.1cm}
\paragraph{Default Delegations.}
For each voter $i$, we let {$d(i) \in \N$} denote their \textit{default delegate}. 
Self-nominations ($d(i)=i$) are allowed and interpreted as abstentions from nominating a default delegate. 
Each \textit{delegation profile} $\de=(d(i))_{i \in V}$ naturally corresponds to a (directed) \emph{delegation graph} $G_{\de} = (\N, \{(i,d(i)) \mid i\in \N\})$ whose edges correspond to default delegations. Thus, each vertex of $G_{\de}$ has out-degree exactly~$1$ and self-nominations correspond to self-loops. 

\vspace{-0.1cm}
\paragraph{Ultimate Delegates.}
Assume that the set of voters casting a vote is known, and let $X \subseteq V$ denote this set. The default delegations of voters in $X$ are irrelevant. Therefore, to resolve delegations for this election\,---\,that is, to determine which voters in $\N \setminus X$ are ultimately represented by which casting voters in $X$ under $\de$\,---\,it suffices to consider the subgraph of $G_{\de}$ that only contains delegations from non-casting voters. For each non-casting voter $i\in V\setminus X$, we can identify their ultimate delegate by following the (unique) directed walk in this graph starting from~$i$. If this walk leads to a casting voter $j \in X$, then $j$ is the ultimate delegate of~$i$. If the walk leads to a cycle or a self-loop,\footnote{
We use the term ``cycle'' exclusively for closed paths involving \textit{at least two vertices}, excluding self-loops from this definition.} then $i$ has no ultimate delegate. 
Each casting voter has a voting weight in the 
examined election equal to the number of voters for whom they are the ultimate delegate, themselves included.\looseness=-1

\vspace{-0.1cm}
\paragraph{Probabilistic Participation.}
A crucial ingredient of our model is the assumption that voters do not know which other voters are casting a vote in the future election. 
Rather, when choosing a default delegate, they need to consider different possibilities of where their vote ``ends up''. To capture this uncertainty, we use a simple probabilistic model. 
Namely, we assume that each voter $i \in \N$ casts a vote with a fixed probability $p_i \in [0,1]$.
We let $\p=(p_i)_{i \in \N}$ denote the profile of vote-casting probabilities. A voter $i$ with $p_i\in \{0,1\}$ is called \emph{deterministic}. 
This modeling choice reflects the idea that, in practice, many (independent) elections occur, and voters are not aware of or engaged in every single one of them. The probability $p_i$ then can be thought of as the fraction of elections in which voter~$i$ typically participates. 
This probabilistic approach provides a simple way to model varying voting behavior without requiring complex assumptions about individual awareness or decision-making for each election.

Given a delegation profile $\de$ and a voter $i \in \N$, we can now derive the probability distribution over $i$'s ultimate delegates. Let $\pi(\de, i)$ denote the longest \textit{simple path} in $G_{\de}$ starting at $i$.
Then, $\pi(\de, i)$ is the unique sequence $(y_0, y_1, \cdots, y_k)$ of distinct voters starting with $y_0=i$ and satisfying the following:

\begin{enumerate}
    \item[(i)] $d(y_{\ell -1})=y_\ell$
    and $y_\ell$ 
    $\notin \{y_0,y_1, y_2, \cdots, y_{\ell-1}\}$ for $\ell \le k$,
\item[(ii)] $d({y_k}) \in \{y_0,y_1, y_2, \cdots, y_k\}$. 
\end{enumerate}

\noindent Voter $i$'s ultimate delegate is the first casting voter along the path $\pi(\de, i)$. 
Therefore, for $\ell \in \{0,1,\cdots,k\}$, the probability that voter $y_\ell$ is the ultimate delegate of $i$ is given by 

\begin{equation} \label{eq:prob}p_{y_\ell}\cdot\prod_{r=0}^{\ell-1} (1-p_{y_r}).
\end{equation}
The ultimate delegate of $i$ is undefined with probability $\prod_{r=0}^{k} (1-p_{y_r})$, which we interpret as $i$'s ballot being lost.

\paragraph{Distance and Tolerance.}
 To evaluate and compare different delegation options, we assume that each voter's utility 
 depends on the alignment between their preferences and those of their ultimate delegate. Alignment is defined in terms of the Euclidean distance between voters along a one-dimensional ideological space, represented as the interval $[0, 1]$. Each voter $i \in \N$ is associated with a fixed position $x_i \in [0,1]$, reflecting their ideological stance, which remains constant throughout all elections. Let $\x=(x_i)_{i \in \N}$ denote \emph{positions}. The utility of a voter decreases as the distance between their position and that of their ultimate delegate increases, capturing the notion that voters prefer representatives who are ideologically closer to themselves.   
 Essentially, a voter’s utility is single-peaked \citep{elkind2022preference} at their own location, decreases linearly within a fixed radius, and drops to zero beyond that threshold. We view this as a natural first modeling step and refer to \cref{sec:concl} for extensions.
 
Each voter $i \in \N$ is associated with a tolerance parameter $\tau_i \ge 0$. 
This represents the maximum distance the voter is willing to accept between their own position and that of their ultimate delegate, while still deriving positive utility from delegating. 
Let  $\text{dist}(i,j)=|x_i-x_j|$ be the distance between voters $i$ and $j$, and let $\Btau = (\tau_i)_{i\in \N}$.
The \textit{acceptability set} of voter $i$ is given by:
$\mathcal{A}_i(\x,\Btau)= \{j \in \N\mid \text{dist}(i,j) \leq \tau_i\}.$

\vspace{-0.1cm}
\paragraph{Instances.} Given the voters' positions $\x$, their voting probabilities $\p$, and their tolerance parameters $\Btau$, we define an instance as the triple $\inst = \langle \x, \p, \Btau\rangle$. To avoid ties, we will always assume that our instances are in \emph{general position}, meaning that no two voters share the same position and that no voter $j \neq i$ is at distance exactly $\tau_i$ from voter~$i$.\footnote{If an instance is not in general position, then slightly perturbing some entries of $\x$ will bring the instance into general position.}
Further, we will sometimes assume that 
$j$ belongs to $\mathcal{A}_i(\x,\Btau)$ if and only if $i$ belongs to $\mathcal{A}_j(\x,\Btau)$.
If this holds for all pairs of voters, then we say that the instance satisfies \emph{mutual acceptance}. A special case of mutual acceptance instances are those where all voters have identical tolerance parameters ($\tau_i=\tau_j$ for all $i,j \in \N$); we will refer to such instances as \emph{symmetric}. 
We often assume that voters are ordered by increasing $x_i$, and then specify $\p$ and $\Btau$ as vectors corresponding to that order.

\vspace{-0.1cm}
\paragraph{Voter Utility.}
The utility of a voter from a realization of the random election is defined as their tolerance minus the distance to their ultimate delegate, or $0$ if the ultimate delegate is undefined. In other words, voters rank potential ultimate delegates by proximity and prefer abstaining over delegating to a voter outside their acceptability set. Formally, given positions $\x$, tolerances $\Btau$, and the set $X$ of casting voters, the utility of voter $i$ under a delegation profile $\de$ is $\tau_i - \text{dist}(i, j)$, where $j \in X$ is the ultimate delegate of $i$, or $0$ if no ultimate delegate is defined. 
Due to probabilistic participation, the ultimate delegate of voter $i$ is a random variable (distributed according to \Cref{eq:prob}). To account for this randomness, we define the \textit{expected utility} of voter $i$ as the weighted sum of their utilities over all possible ultimate delegates. Specifically, the expected utility of voter~$i$ (henceforth, simply \emph{utility}) in an instance $\inst$
can be expressed as

\begin{equation}
u_i(\de,\inst)  = \sum_{\ell=0}^k (\tau_i - \text{dist}(i, y_{\ell})) \cdot p_{y_{\ell}} \cdot \prod_{r=0}^{\ell-1}(1-p_{y_r})
\label{eq:utility}
\end{equation}
where $(y_0, y_1, \cdots, y_k)$ are the voters along the path $\pi(\de,i)$.  
When the instance $\inst$ is clear, we refer to $u_i(\de,\inst)$ as $u_i(\de)$. 
Normalization would not affect our results on NE existence and structure.
The \emph{social welfare} of a profile $\de$ in $\inst$ is the sum over voter utilities, 
$SW(\de,\inst) = \sum_{i\in \N}u_i(\de,\inst)$. The profile maximizing it among all possible profiles for $\inst$ will be called \emph{optimal} and denoted by $\de_{SW}(\inst),$ or simply $\de_{SW}$.

\begin{example}\label{Longterm:ex:start}
Consider an instance with six voters, $\N=\{A,B,C,D,E,F\}$. 
The positions~$\x$ and voting probabilities~$\p$ are depicted below, together with the delegation graph $G_{\de}$ of the delegation profile $\de$ with
$d(A)=B$,
$d(B)=C$,
$d(C)=A$,
$d(D)=E$,
$d(E)=D$,
$d(F)=F$.

\begin{figure}[h!]
\centering
\resizebox{0.95\columnwidth}{!}{
\begin{tikzpicture}[node distance=1cm,semithick, scale=1]
\node at  (2,0.5) {\small $0.2$};
\node at  (3,0.5) {\small $0.3$};
\node at  (4,0.5) {\small $0.4$};
\node at  (5,0.5) {\small $0.5$};
\node at  (6,0.5) {\small $0.6$};
\node at  (8,0.5) {\small $0.8$};
\node (A) at  (2,-0.3) {$A$};
\node (B) at  (3,-0.3) {$B$};
\node (C) at  (4,-0.3) {$C$};
\node (D) at  (5,-0.3) {$D$};
\node (E) at  (6,-0.3) {$E$};
\node (F) at  (8,-0.3) {$F$};

\node at  (0,-0.3) {$\de:$};
\node at  (0,0.5) {$\x:$};
\node at  (0,1) {$\p:$};
\node at  (2,1) {\small $0.8$};
\node  at  (3,1) {\small $0.3$};
\node at  (4,1) {\small $0.2$};
\node at  (5,1) {\small $0.3$};
\node at  (6,1) {\small $0.1$};
\node at  (8,1) {\small $0.3$};

%
%

\path[thick] (1,0.2) edge (8.5,0.2)
            (2,0.2) edge (2,0.15)
            (3,0.2) edge (3,0.15)
            (4,0.2) edge (4,0.15)
            (5,0.2) edge (5,0.15)
            (6,0.2) edge (6,0.15)
            (8,0.2) edge (8,0.15);
%
%
\draw[->] (F) to [loop left]  (F);

\draw[->]           (A) edge (B)
                        (B) edge (C)
                        (C) edge[bend left] (A)
                        (D) edge[bend left] (E)
                        (E) edge[bend left] (D);
%
\end{tikzpicture} }
\label{fig:longtermstart}
\end{figure}
Suppose that for a given election, the set of casting voters is $\{A,F\}$.
This situation happens with probability 
$p_A\cdot p_F \cdot \Pi_{i\in \{B,C,D,E\}} (1-p_i)= 0.8 \cdot 0.3 \cdot 0.7 \cdot 0.8 \cdot 0.7 \cdot 0.9  \approx 0.085.$
In this scenario, $A$ has a voting weight of $3$ and $F$ has a voting weight of $1$. The voting weights of $D$ and $E$ cannot be allocated to a casting voter and are lost.

Assume that $\tau_i=0.25$ for all $i \in \N$. Then the acceptability set of voter $D$ is $\mathcal{A}_D(\x,\Btau)=\{B,C,D,E\}$.
The expected utility of 
voters $A$ and $F$ with respect to the delegation profile $\de$ are 
$u_A(\de)=0.8\cdot0.25+0.2\cdot 0.3\cdot (0.25-0.1)+0.2\cdot 0.7 \cdot 0.2\cdot (0.25-0.2)\approx 0.21$ and $u_F(\de)=0.3\cdot 0.25=0.075$.
\end{example}

\section{Existence of Nash Equilibria}
\label{sec:existence}

We start our game-theoretic analysis of the default delegation model. 
We define the concept of best responses and profitable deviations in a standard way. 
We denote by $\de_{-i}$ the profile $\de$ not including the choice of $i$, and by $(\de_{-i},d'(i))$ the delegation profile in which all voters except $i$ delegate according to $\de$, whereas $i$ delegates to $d'(i)$. 

\begin{definition}\label{def:BR}
For a voter $i\in \N,$ $d'(i)$ is a \emph{best response} to delegation profile $\de$ if and only if 
it maximizes $u_i(\de_{-i},\cdot)$.
We say that $d'(i)$ is a \emph{profitable deviation} from $\de$ for voter $i$ if 
$u_i(\de_{-i},d'(i)) > u_i(\de)$.
\end{definition}

Building on the notion of profitable deviations, we are ready to define (pure) Nash equilibria. We focus on pure (rather than mixed) equilibria, 
not only because they are the standard starting point for strategic reasoning on a novel model, but also
because practical implementations of liquid democracy typically require deterministic delegations, making pure equilibria an appropriate model of real-world delegation behavior.

\begin{definition}\label{def:NE}
A delegation profile $\de$ is a \emph{Nash equilibrium (NE)} if no voter $i\in\N$ has a profitable deviation from $\de$. 
\end{definition}

Through our initial example, we highlight that equilibria are not necessarily unique and that different equilibria may have different graph-theoretic structures. 

\begin{example1cont}\label{ex:NE}

The delegation profile $\de$ from \Cref{Longterm:ex:start} is a Nash equilibrium, since each voter chooses a best response, as shown in 
\Cref{tab:NEexample}.
Note that equilibrium delegations need not go to the closest voter.
Interestingly, $\de$ is not the only NE of this instance. It can be verified that $\de'$, whose delegation graph $G_{\de'}$ is shown below, is also an NE. 
We highlight that two equilibria might significantly differ. For instance, $G_{\de'}$ has two (weakly) connected components, in contrast to $G_{\de}$. Moreover, $\de$ and $\de'$ differ in terms of social welfare, casting voters' weights and expected number of votes that are lost. 

\begin{figure}[h]
\centering
\resizebox{0.9\columnwidth}{!}{
\begin{tikzpicture}[node distance=1cm,semithick,scale=1]

\node (A) at  (2,-0.3) {$A$};
\node (B) at  (3,-0.3) {$B$};
\node (C) at  (4,-0.3) {$C$};
\node (D) at  (5,-0.3) {$D$};
\node (E) at  (6,-0.3) {$E$};
\node (F) at  (8,-0.3) {$F$};
\node (g) at  (1,-0.3) {$\de'$:};

\draw[->] (F) edge [bend left]  (E)
                (E) edge[bend left](F)
                (D) edge[] (C)
                (C) edge[] (B)
                (B) edge[] (A)
                (A) to[bend right] (C);

\end{tikzpicture}}%
\end{figure}

\begin{table}[t]
    \vspace{0.3cm}\centering
    \[
\scalebox{0.88}{\begin{tabular}{ccccccc}
    \toprule
      &$A$&$B$&$C$&$D$&$E$&$F$  \\
     \midrule
    $A$&0.200 &  $\mathbf{0.210}$  & 0.202 & 0.195 & 0.194 & 0.179\\
    $B$&0.159 & 0.075 &  $\mathbf{0.163}$  & 0.083 & 0.081 & 0.023\\
    $C$& $\mathbf{0.089}$  & 0.086 & 0.050 & 0.089 & 0.086 & 0.014\\
    $D$&0.052 & 0.086 & 0.075 & 0.075 &  $\mathbf{0.085}$  & 0.064\\
    $E$&-0.084 & -0.043 & -0.055 &  $\mathbf{0.066}$  & 0.025 & 0.039\\
    $F$&-0.134 & -0.102 & -0.111 & 0.067 & 0.069 &  $\mathbf{0.075}$ \\
    \bottomrule
\end{tabular} }
\]%
    \caption{Expected utility for deviations from $\de$ in \Cref{Longterm:ex:start}. The entry in cell $(i,j)$ corresponds to $u_i(\de_{-i},j)$ and the entries corresponding to best responses are indicated in bold.}
\label{tab:NEexample}
\end{table}
\end{example1cont}

\paragraph{Experimental Analysis of NE Existence.} To get a first impression of whether Nash equilibria exist in general, we carried out computational experiments using a \emph{best-response dynamic}. That is, the process starts with some delegation profile (e.g., a random one)
and then iterates over the voters, updating their delegation whenever there is a profitable deviation. The process stops when no voter can make a profitable deviation, which results in a Nash equilibrium by definition. Interestingly, running our best-response dynamics on 20,000 different (mutual acceptance) instances for various values of $n$, $\x$, $\p$, and $\Btau,$ and starting profiles, has always led to the identification of a Nash equilibrium. The details are deferred to the Appendix.\looseness-1

\smallskip \noindent \textbf{Non-Existence of Equilibria.} In contrast to what we observed in the computational experiments sketched above, Nash equilibria do not always exist in the default delegation model. To showcase this, we provide the following example
and later strengthen the result in \Cref{thm:no-NE}. 

\begin{example}
\label{ex: no-NE}
    The instance $\inst$ with $\N=\{A, B, C, D\}$,
$\x=(0, 0.05, 0.1, 0.5)$,
$\p=(0.4, 0.05, 0.2, 0.4)$ ,
$\Btau=(1, 0, 0.2, 0)$
does not admit an NE.
    Assume for contradiction that~$\d$ is an NE in $\inst$. Then, in $\d$, the voters $B$ and $D$ must delegate to themselves as $\tau_B=\tau_D=0$, i.e., delegating to any other voter lowers their utility compared to self-delegation. Further, since $\tau_C= 0.2,$ $\text{dist}(C,D)=0.4$ and $d(D)=D,$ voter $C$ does not delegate to $D$. Neither $A$ nor $C$ delegates to themselves. This is because $\tau_A - \text{dist}(A,B) >0$, $\tau_C - \text{dist}(C,B) >0$, and $d(B)=B$, implying that choosing to delegate to $B$ provides better utility than self-delegation. In the following table, we show the utilities of $A$ and $C$ in all profiles that were not ruled out by the previous reasoning.

 \vspace{-0.3cm}
 \begin{figure}[h]
 \centering
\resizebox{0.6\columnwidth}{!}{
\nfgametwothree{$d(C)=A$ $d(C)=B$ $d(A)=B$ $d(A)=C$ $d(A)=D$ .075 .428 .072 .508 .033 .52 .044 .428 .044 .53 .044 .52}}
\vspace{-0.3cm}
\end{figure}

\noindent The table depicts the normal form representation of the game induced by $\inst,$ 
where rows correspond to the possible choices of voter $C$ and columns to those of voter $A,$ for specifying $\de$. Observe that no possible delegation profile is an NE.

\end{example}

\noindent Given this impossibility result, we focus on subclasses of the default delegation model or slight variations of it, in order to obtain positive results on the existence of Nash equilibria. 

\subsection{Special Cases}
We discuss three special cases of the default delegation model and draw a complete picture of whether these restrictions suffice to guarantee the existence of equilibria. We study
(i) \textit{deterministic} instances, i.e., those  where $p_i \in \{0,1\}$ for all $i \in V$, 
(ii) \textit{mutual acceptance} instances, i.e., those  where $j \in \mathcal{A}_i(\x,\Btau)$ if and only if $i \in \mathcal{A}_j(\x,\Btau)$ for all $i,j \in V$, and (iii) instances with \textit{few voters}, i.e., those where $|V|$ is upper bounded by a constant. 
On the positive side, such restrictions can guarantee the existence of an NE. 

\begin{theorem} \label{thm:NE-existence}
    For each of the following restrictions, any instance $\inst$ is guaranteed to admit a Nash equilibrium, which is moreover computable in polynomial time: 
    \begin{enumerate}[label=(\roman*),leftmargin=*,topsep=2pt,itemsep=1pt,parsep=0pt]
        \item $\inst$ is deterministic, 
        \item $\inst$ has two voters, i.e., $|V| =2$, or
        \item $\inst$ satisfies mutual acceptance and $|V| \leq 3$.
    \end{enumerate}
\end{theorem}

In deterministic instances, the profile where every non-casting voter (i.e., with $p_i=0$) delegates to their closest casting voter (i.e., with $p_j=1$) in their acceptability set is an NE. 
For two-voter instances, we show that the expected utility of a voter is not influenced by the delegation choice of the other. 
For mutual acceptance instances with three voters, we propose a greedy algorithm for finding an NE. 
While these restrictions are strong, we complement \Cref{thm:NE-existence} by showing that relaxing them even slightly invalidates the result. 

\begin{theorem} \label{thm:no-NE}
    There exists an instance $\inst$ for which no Nash equilibrium exists if 
    \begin{enumerate}[label=(\roman*),leftmargin=*,topsep=2pt,itemsep=1pt,parsep=0pt]
        \item $\inst$ has three voters, i.e., $|V| = 3$, or
        \item $\inst$ satisfies mutual acceptance and $|V| =4$. 
    \end{enumerate}
\end{theorem}

We remark that statement (ii) of \cref{thm:no-NE} holds even for \emph{symmetric} instances, i.e., with voters of equal tolerance.

\subsection{Variants of the Model}

In response to these negative results, we discuss two variants of our model that guarantee the existence of NE. 
These are relevant in certain real-world contexts, but mainly serve a theoretical purpose: demonstrating the tightness and minimality of \cref{ex: no-NE} and complementing \cref{thm:no-NE}.

\medskip \noindent \textbf{Directional Acceptability.} \label{sec:left-right}
In the default delegation model, a voter accepts representation by voters positioned both to their left and to their right, only dependent on their distance. We introduce a variant of the model, where each voter selects a direction and accepts representation only by voters in that direction,
in which case the utility is still determined by the distance, in line with \cref{eq:utility}. 
Depending on the direction selected, we refer to a voter as \emph{leftist} or \emph{rightist}. 
To motivate the setting further, we note that beyond the political spectrum the positional line can instead represent a hierarchy or organizational seniority, where participants may delegate only upward and, e.g., set a personal cutoff allowing delegation to a supervisor or manager but not to executive levels.

In \cref{ex: no-NE}, voter $A$ can be considered a rightist
and
voter $C$ can be considered a leftist as $\mathcal{A}_C(\inst)$ only includes voters to their left.
Hence, NE may not exist in instances with both leftists and rightists.
However, \Cref{thm: NE-LR} shows that any instance with only leftists or only rightists admits an~NE.

For the sake of concreteness, we define an example of a utility function that induces leftist voters. Namely, replace $(\tau_i - \text{dist}(i, y_{\ell}))$ in \Cref{eq:utility} (intuitively, the utility of voter $i$ for being represented by $y_{\ell}$ in a specific election) by:
 
\[ \begin{cases} \tau_i - \text{dist}(i,y_{\ell}), & \text{ if } y_{\ell} \text{ is left of } i, \text{ i.e., $x_{y_\ell}<x_i,$}\\ 
- \text{dist}(i,y_{\ell}), & \text{ if } y_{\ell} \text{ is right of } i, \text{ i.e., $x_{y_\ell}>x_i$.}
\end{cases}\]

\noindent We remark that \Cref{thm: NE-LR} holds for any utility model that assigns negative utility to representation on the one side and utility equal to $\tau_i - \text{dist}(i,y_{\ell})$
to the other.

\begin{theorem}\label{thm: NE-LR}
Every instance in which the voters are all leftists or all rightists admits a Nash equilibrium, computable in polynomial time.
\end{theorem}

The proof of \Cref{thm: NE-LR} constructs a Nash equilibrium by starting from a profile where everyone delegates to themselves, and then finding best responses for all voters sequentially in order of their position.

\medskip \noindent \textbf{Proxy Voting.}
We now move to another variant of the model in which Nash equilibria are guaranteed to exist. 
In the \emph{proxy voting} setting, we restrict the number of voters on any path $\pi(\de, i)$ that leads to a casting voter. Specifically, no such path is allowed to contain more than two voters (including voter $i$ themselves). 
Hence, we effectively restrict the strategy space of the voters based on the actions of the other voters. 
This restriction is reminiscent\footnote{The settings are not identical as we allow for length 2 cycles.} of the well-established framework of proxy voting \citep{CMM+17an,anshelevich2021representative}, 
a variant of liquid democracy in which voters are divided into delegating and casting voters and delegation chains may contain at most $2$ voters. This setting aligns with implementations in blockchain governance and delegated Proof-of-Stake systems, where voters delegate to designated actors (e.g., active voters, token-holder representatives, or validators) who, by design, cannot delegate further, resulting in delegation chains of minimal length.

In \cref{ex: no-NE}, delegation chains of three voters arose. By forbidding such chains, we effectively eliminate the issue that leads to the non-existence of equilibria. In the proxy voting setting, we guarantee the existence of Nash equilibria, leading to a dichotomy in the maximum allowable delegation chain length to ensure the existence of Nash equilibria.

\begin{theorem}\label{thm: NE-proxy}
In the proxy voting setting, every instance admits a Nash equilibrium, computable in polynomial time.
\end{theorem}

\section{Structure of Nash Equilibria}
\label{sec:structure}

We now focus on the structural properties of equilibria. In particular, we are interested in the existence of cycles in delegation graphs corresponding to Nash equilibria. 
Our first result establishes that delegation cycles are the rule, rather than the exception. This aims to provide a compelling game-theoretic explanation for the prevalence of cycles and the behavior of voters observed in practice.

\begin{theorem}
    \label{obs:NEcycles}
    
        Consider a mutual acceptance instance $\inst$ without deterministic voters.
        Then, for every Nash Equilibrium $\de$ of $\inst$, it holds that every weakly connected component of $G_{\de}$ with more than a single vertex 
    has exactly one cycle. 
\end{theorem}

When the assumptions of \cref{obs:NEcycles} do not hold, cycles do not necessarily exist in every equilibrium.

\begin{observation}
\label{prop:degenerate}
    Cycles are not guaranteed to exist in $G_{\de},$ where $\de$ is an NE of an instance that is not of mutual acceptance or where deterministic voters exist.
\end{observation}

However, in mutual acceptance instances, at least one equilibrium with a cyclic structure is guaranteed to exist, even with some deterministic voters. The proof is similar to that of \cref{obs:NEcycles} and includes the observation that deterministic voters may be indifferent between delegation options.

\begin{theorem}
    \label{obs:NEcycles_degenerate}
    Consider a mutual acceptance instance $\inst$ admitting an NE. Then, there exists an NE  $\de$ of $\inst$ in which every weakly connected component of $G_{\de}$ with more than a single vertex has exactly one cycle. 
\end{theorem}

Without mutual acceptance, the existence of equilibria exhibiting cycles is not guaranteed. For example, consider an instance with two non-deterministic voters such that $A$ accepts $B$, but $B$ does not accept $A$. Then, there is a unique equilibrium in which $A$ delegates to $B$ and $B$ self-loops.

\smallskip

Returning to the case where the assumptions of \Cref{obs:NEcycles} hold, we now aim to further analyze the structure of equilibria by turning our attention to delegations ``entering'' a cycle. Specifically, for a weakly connected component $W$, let $\mathcal{C}(W)$ denote the set of voters forming the cycle within that component, and let $\mathcal{L}(W)$ and $\mathcal{R}(W)$ denote the sets of voters of $W$ positioned to the left and right of the cycle, respectively. Formally, $\mathcal{L}(W) = \{i \in W: x_i < x_j \text{ for all } j \in \mathcal{C}(W)\}$ and
$\mathcal{R}(W) = \{i \in W: x_i > x_j \text{ for all } j \in \mathcal{C}(W)\}$.

\begin{theorem}\label{obs:entrypoints}
Consider a mutual acceptance instance $\inst$ without deterministic voters and a Nash equilibrium $\de$ of $\inst$.
Consider a weakly connected component $W$ of $G_{\de}$ that consists of more than a single vertex, and let $\mathcal{C}(W)$ denote the cycle in $W$. 
There is at most one vertex $v_L \in \mathcal{L}(W)$ with $d(v_L) \in \mathcal{C}(W)$ and at most one vertex $v_R \in \mathcal{R}(W)$ with $d(v_R) \in \mathcal{C}(W)$.
    Moreover, in $G_{\de}$, $\mathcal{L}(W)$ and $\mathcal{R}(W)$ form in-trees rooted at $v_L$ and $v_R$, respectively.
\end{theorem}

Thus, the cycle $\mathcal{C}(W)$ has a unique ``entry point'' $v_L$ for voters in $\mathcal{L}(W)$, and all voters in $\mathcal{L}(W)$ have delegation paths to $v_L$ (analogously for $v_R$ and $\mathcal{R}(W)$). 
It might be tempting to conjecture that these entry points $v_L$ and $v_R$ delegate to the leftmost and rightmost voters in $\mathcal{C}(W)$, respectively, or that all voters in  $\mathcal{L}(W)$ (respectively, $\mathcal{R}(W)$) form a simple delegation path. 
However, in the Appendix, we show that this is not generally the case. 
Therefore, a significant strengthening of the structural description offered by \cref{obs:entrypoints} is unlikely.

\smallskip \noindent \textbf{Experimental Analysis of NE Structure.} 
Our theoretical results do not specify how large delegation cycles are or how often they occur in instances not satisfying the assumptions of \Cref{obs:NEcycles}. We address these questions experimentally in general instances (see the Appendix for details). 
We examined the \textit{size} (i.e., number of vertices) and \textit{width} (i.e., maximum distance between two vertices) of cycles and weakly connected components and we observed that, as tolerance levels decrease, cycle size and width, as well as component width, decline gradually. 
Moreover, as $n$ increases, the average cycle and component width decreases, with voters in the same component\,---\,especially cycles\,---\,having closely aligned positions. The proportion of voters with self-loops remains stable at around 5\%. The average cycle size stays around $4.9$ across instances and grows as $n$ increases. 
Notably, nearly all weakly connected components with more than one vertex contain a cycle, indicating that the pattern identified theoretically for mutual acceptance instances (see \cref{obs:NEcycles}) also appears in randomly generated instances.

\section{Quality of Nash Equilibria}
\label{sec:quality}
\label{sec:qualityofNE}
We next focus on evaluating the quality of equilibria. For the theoretical analysis we use a Price-of-Anarchy approach to compare the social welfare of Nash equilibria to the optimal social welfare \citep{PoA}.
First, we observe a contrast between the structure of social-welfare-maximizing delegation graphs and of Nash equilibria.

\begin{observation}
\label{SCno-cycle}
    There exist mutual acceptance instances without deterministic voters in which 
    social welfare maximizing delegations do not induce cycles; e.g., in
    the symmetric instance with $\N=\{A,B,C\}$, $\x=(0.12, 0.5, 0.88)$, $\p=(0.1, 0.9, 0.1)$, and  
    $\tau_i=0.4$ for all $i \in \N$, social welfare is maximized if 
    $A$ and $C$ delegate to $B$, who self-loops. 
\end{observation}

Recall that $\de_{\texttt{SW}}(\inst)$ is a profile maximizing social welfare and let $\de_{\texttt{NE}}(\inst)$ be a Nash equilibrium of $\inst$ achieving the \textit{lowest} social welfare among all NE. 
We define the  \emph{Price of Anarchy (PoA)} of an instance $\inst$ in the standard way

\[\textit{PoA}(\inst)=\frac{SW(\de_{\texttt{SW}}(\inst))}{SW(\de_{\texttt{NE}}(\inst))}\]
and show that this ratio can be arbitrarily large.

\begin{theorem}\label{thm: poa}
The Price of Anarchy 
of default delegation instances 
is unbounded.
\end{theorem}

At first glance, \cref{thm: poa} is a negative result concerning the quality of NE. However, the constructed instances have certain characteristics, such as voters with a very low voting probability 
and all but one voter having acceptability sets limited to themselves, while voter $1$ has $\mathcal{A}_1(\x,\Btau) = \N$.
Moreover, the social welfare values of the two delegation profiles in the proof of \cref{thm: poa} exhibit a relatively small absolute difference.
This suggests that measuring the difference rather than the ratio 
may yield more informative conclusions about the quality of equilibria. 
We define the \emph{additive Price of Anarchy} of an instance $\inst$ as 
$\textit{PoA}^+(\inst) =SW(\de_{\texttt{SW}}(\inst))-SW(\de_{\texttt{NE}}(\inst)),$
and 
proceed with the following positive results on both the multiplicative and the additive Price of Anarchy.\footnote{While $\textit{PoA}^+$ is not explicitly normalized by utility, the bound in \Cref{thm:poa upper} implicitly is, since $\sum_{i \in [n]} \tau_i$ is a trivial upper bound for social welfare.}
\begin{theorem}
\label{thm:poa upper}
    For every instance $\inst$,
    $\textit{PoA}(\inst) \leq \nicefrac{1}{p_{\min} }$ and $\textit{PoA}^+(\inst) \leq (1{-}p_{\min})\sum_{i\in \N}\tau_i,$ where $p_{\min}{=}\min_{i\in \N}\{p_i\}$.
\end{theorem}

\cref{thm:poa upper} asserts that
higher voting probabilities correlate with better Nash equilibria in terms of social welfare. Furthermore, and perhaps surprisingly, the smaller the tolerance levels, the better the additive PoA bound.

We have assessed the expected number of votes cast in equilibria, showing that, unlike in other liquid democracy frameworks where cycles are criticized for resulting in ballot loss, in our setting, they effectively help mitigate lost voting power.
Also, results on the structure of optimal delegation profiles, complementing \cref{SCno-cycle}, can be found in the Appendix, further highlighting differences from NEs and vote-loss-minimizing profiles.

\smallskip \noindent \textbf{Experimental Analysis of NE Quality.}\label{sec:PoA-exp}
To complement our worst-case bounds, we move beyond PoA and examine how the social welfare achieved by NEs (in particular under the best-response dynamic) compares to the optimal social welfare in randomly generated instances. Since identifying $\de_{SW}$ is computationally infeasible when $n$ is large, we approximate it by the sum of each voter’s expected utility under their optimal delegation profile, denoted by $\ODP(\inst)$. Formally, $\ODP(\inst) = \sum_{i \in \N} u_i(\boldsymbol{d^{i*}})$, where $\boldsymbol{d^{i*}}$ is a delegation profile maximizing the utility of voter $i$. It holds that $G_{\boldsymbol{d^{i*}}}$ contains a path that starts in $i$ and passes through all vertices in $\mathcal{A}_i(\x,\Btau)$ in increasing order of distance to $i$. 
This value serves as an upper bound of $SW(\de_{SW})$ and does not necessarily correspond to a feasible delegation profile.

\newcommand{\dne}{\de_{\mathit{BR}}}

We generated 100 instances for each $n \in \{20, 50, 100, 200\}$, with  values for $\x,\p,\Btau$ chosen uniformly at random. For $n=50$, we also tested 5 tolerance vectors $\Btau$, scaling each by 0.75 and 0.5 to assess the effect of different tolerance levels (details on the instances can be found in the Appendix).
For each instance, we computed $\ODP(\inst)$ as an upper bound on the social welfare and an NE $\dne$ via best-response dynamics. \Cref{tab:SW_varyn} shows the average ratios $SW(\dne) / ODP(\inst)$.  As a baseline, we also include the social welfare achieved by the delegation profile $\d_{\texttt{dir}}$ (``direct voting'') in which every voter self-loops.

\begin{table}[t]
    \centering
     \resizebox{.999\columnwidth}{!}{
\begin{tabular}{cccccccc}
   \toprule
    & \multicolumn{4}{c}{Number of Voters} & \multicolumn{3}{c}{$\tau_i \in [0,\tau_{\max}]$ with $\tau_{\max} =$} \\
    \cmidrule(lr){2-5} \cmidrule(lr){6-8}
    & 20  & 50  & 100  & 200 & 1 & 0.75 & 0.5 \\
    \midrule
    $\dne$  & 97.6\%  & 98.8\%  & 99.4\%  & 99.7\%  & 98.9\% & 98.9\%  & 98.6\%  \\
    $\de_{\texttt{dir}}$  & 53.6\%  & 49.6\%  & 50.4\%  & 50.5\% & 51.2\%  & 51.6\%  & 52.3\%  \\
    \bottomrule
\end{tabular}
 }
    \caption{The average social welfare achieved by $\dne$ and $\de_{\texttt{dir}}$ in our experiments, as a percentage of $\ODP(\inst)$. 
    }
    \label{tab:SW_varyn}
\end{table}

As expected, the Nash equilibrium profiles
outperform the direct voting profiles, which consistently 
reach only around 50\% of  
$\ODP(\inst)$.
The average social welfare achieved by
$\dne$ is remarkably high ($\ge$98\% of  
$\ODP(\inst)$)
and gets closer to $\ODP(\inst)$ as $n$ increases. 
Given that $\ODP(\inst)$ upper bounds the optimal social welfare, we conclude that Nash equilibria in our model achieve almost optimal social welfare.

\section{Conclusion}
\label{sec:concl}
We introduced the default delegation model and used it to provide a novel game-theoretic perspective on strategic delegation decisions in liquid democracy.
We revealed how delegation cycles naturally emerge among rational participants, offering a justification for their existence. 

Our model leads to several avenues for future research. One of them is to explore the computational complexity of finding  equilibria or delegation profiles maximizing social welfare. In our experiments, we use best-response dynamics to find equilibria; however, these algorithms are not guaranteed to converge. 
It would also be interesting to define voters' alignment based on more general metric spaces, e.g., the Euclidean plane. 
Preliminary experiments reveal that the structure of equilibria becomes more complicated in that setting. 
Considering alternative utility functions could yield further insights. Examples include normalized utility, non-linear functions of distance, and functions incorporating voting costs, voting behavior of the ultimate delegate, or the resulting election outcome. 
The probability of casting a ballot could also be part of the strategy space of a voter.
Finally, while we primarily focused on cycles, \textit{long delegation paths} remain an important and underexplored aspect of liquid democracy, as they may be associated with eroding trust in ultimate delegates.

\section*{Acknowledgments}
 We thank Andreas Nitsche for helpful discussions. 
 Rachael Colley has been supported by the ANR JCJC project SCONE (ANR 18-CE23-0009-01) and the grant EP/X013618/1 from the Engineering and Physical Sciences Research Council. 
 This project has received funding from the European Research Council (ERC) under  the European Union’s Horizon 2020 research and innovation programme (grant  agreement No 101002854). Grzegorz Lisowski acknowledges support by the European Union under the Horizon Europe project Perycles  (Participatory Democracy that Scales).
 Georgios Papasotiropoulos was supported by the European Union
 (ERC, PRO-DEMOCRATIC, 101076570). Views and opinions expressed are, however, those of the
 authors only and do not necessarily reflect those of the European Union or the European Research
 Council. Neither the European Union nor the granting authority can be held responsible for them. 
 Ulrike Schmidt-Kraepelin was supported by the Dutch Research Council (NWO) under project number VI.Veni.232.254.
 
\hspace{-1.5em}
\includegraphics[width=0.6\columnwidth]{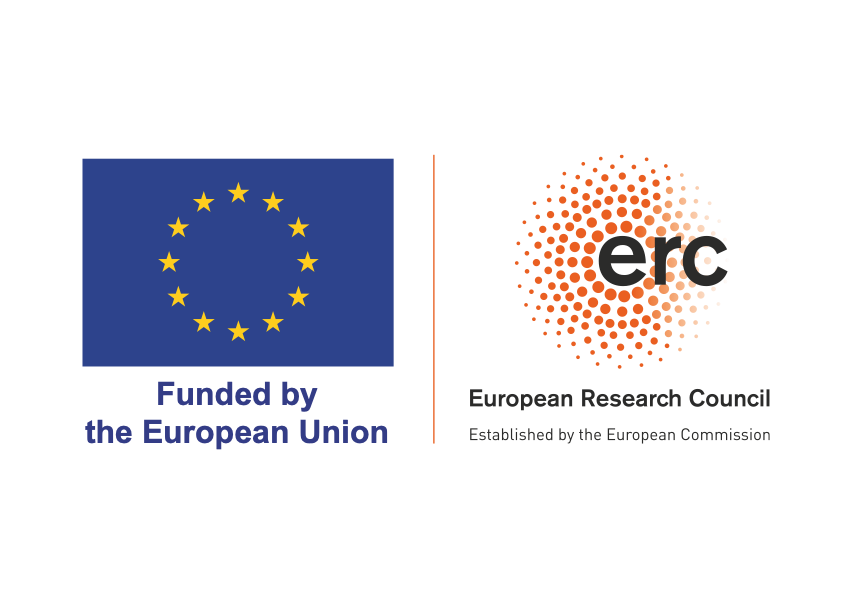}
~
\hspace{-0.2cm}
\raisebox{1.5cm}{\includegraphics[width=0.35\columnwidth]{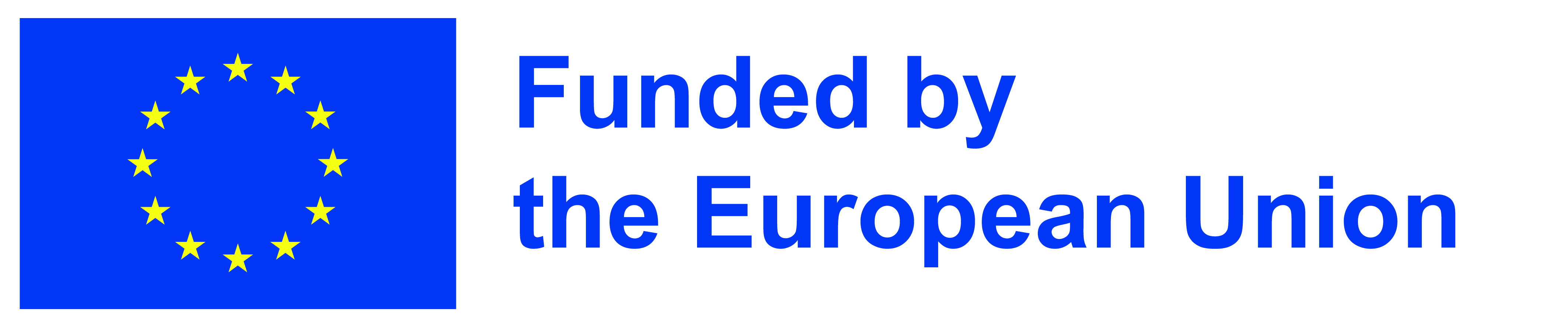}}

{\small{
\bibliographystyle{named}
\bibliography{references}
}}

\newpage

\appendix

\onecolumn
\section{Omitted Proofs}
\label{app:proofs}

\subsection*{Proof of \cref{thm:NE-existence}}

\begin{enumerate}
\item[(i)] Suppose that the voters in $\N$ are split into two disjoint groups: The casting voters, i.e., those in the set $X = \{i\in \N \mid p_i = 1\}$, and the delegators, i.e., those in the set $\N\setminus X = \{i \in \N \mid p_i = 0\}$.  
We set $d(i) = i = \argmin_{j \in \mathcal{A}_i(\x,\Btau) \cap X}(|x_i - x_j|)$ for all voters in $X$. For each $i \in \N\setminus X$, we also set  

$$d(i) = \argmin_{j \in \mathcal{A}_i(\x,\Btau) \cap X}(|x_i - x_j|),$$  

i.e., they delegate to their closest casting voter in their acceptability set. If $\mathcal{A}_i(\x,\Btau) \cap X = \emptyset$, then we simply set $d(i) = i$.  

We proceed to show that $\de$ is a Nash equilibrium. No voter $i \in X$ has a profitable deviation, as their expected utility will always be equal to $\tau_i$ regardless of the delegation they choose, due to $p_i = 1$. For voters $i \in \N\setminus X$, we consider two cases, i.e., when $d(i) = i$ and when $d(i) \neq i$.  

In the first case, namely when $d(i) = i,$ if $i$ deviates, they will finally be represented by a voter in $X,$ receiving a utility of $0,$ or by a voter in $\N\setminus X$. In the latter case they will necessarily be represented by a voter not in 
$\mathcal{A}_i(\x,\Btau)$, because $X \cap \mathcal{A}_i(\x,\Btau)=\emptyset$. Thus, their expected payoff would be non-positive in both cases, and consequently it will be no more than the utility they receive with $d(i) = i$, which gives a utility of $0$.  

In the second case, when $d(i) \neq i$, delegating to any other voter in $\mathcal{A}_i(\x,\Btau) \cap X$ will not increase their expected utility, because $d(i)$ is the closest casting voter to $i$.
Similarly, delegating to a voter $j \in \N\setminus X$ or to a voter $j \in X \setminus \mathcal{A}_i(\x,\Btau)$ might result in $i$ being represented by a voter $k$ satisfying
$k=\argmin_{j \in \mathcal{A}_i(\x,\Btau) \cap X}(|x_i - x_k|),$  
or not, depending on where $j$ delegates. In both cases, this is again not a profitable deviation from  $\de$ for $i$.  
\item[(ii)] We let $\N=\{A,B\}$ and we denote by $d=|x_B-x_A|$. First, we examine the expected utility for voter $A$, assuming that $d(B)=B$. If $d(A)=B$, then $u_A(d(A),d(B))=\tau_Ap_A+(1-p_A)p_B(\tau_A-d)$, and if $d(A)=A$, then $u_A(d(A),d(B))=\tau_Ap_A.$ Crucially, the utility of $A$ remains exactly the same even when $d(B)=A$. Therefore, in an equilibrium, voter $A$ decides to delegate to $B$ if $(1-p_A)p_B(\tau_A-d)\geq0,$ or equivalently if $\tau_A\geq d$, and self-loops otherwise, regardless of what strategy voter $B$ selects. The analogous arguments can determine the delegation of voter $B$ in an equilibrium.
\item[(iii)] We call the set of the three voters $\N=\{1,2,3\}$ and we assume that the voters are named in increasing order of their positions in the line. We can assume that $p_i > 0$ for all $i \in \{1, 2, 3\}$. Otherwise, if there is a deterministic voter $i$ for which this condition does not hold, we do the following: first, we exclude voter $i$ from consideration and then we find a NE in the remaining instance that has at most two voters, which can be done as previously shown. The exclusion of $i$ is safe because it is never profitable for a voter to delegate to $i$ regardless of the choice of voter $i$. Finally, having fixed the delegation for each voter $j$ such that $p_j>0,$ we fix a delegation for $i$ by determining the optimal delegation for them, given the choices of others.

At a high level, the proof for finding a Nash equilibrium $\de$, goes as follows: First, we fix a delegation for voter $2$ according to some criterion and, depending on this choice, we then fix the delegation for a (specific) second voter. Finally, we fix the delegation for the remaining voter as well. We will prove that by following this process and the specified criteria we obtain a delegation profile that is a Nash equilibrium. Specifically, we will show that no voter will have an incentive to deviate after fixing their choice, no matter the choices of the rest.

Let us compute the utility that voter $2$ will get from each possible delegation, assuming that $d(i)=i$ for $i\in \{1,3\}$. Then, for each such $i$, it holds $u_2(1,i,3)=p_2\tau_2+(1-p_2)p_i(\tau_2-\text{dist}(2,i)).$ 
We set $d(2)=\argmax_{i\in \{1,3\}} \{p_i(\tau_2-\text{dist}(2,i))\},$ 
assuming that $\max_{i\in \{1,3\}}\{p_i(\tau_2-\text{dist}(2,i))\}>0$. The case where $\max_{i\in \{1,3\}}\{p_i(\tau_2-\text{dist}(2,i))\}\leq 0$ will be addressed separately later. 

Assume without loss of generality that $\argmax_{i\in \{1,3\}}\{p_i(\tau_2-\text{dist}(2,i))\}=3$ and so that $d(2)=3.$ In that case, we continue with fixing the delegation for voter $1$, and finally for voter $3$ (we reverse the order of examination if $d(2)=1$), and we claim, and we will soon prove, that regardless of those, voter $2$ will not have an incentive to deviate from $d(2)=3$. We begin with examining voter $1$, i.e., the voter that voter $2$ did not delegate to. We compute the utility that voter $1$ would get if $d(2)=3=d(3),$ and based on this, we will fix $d(1)$. It holds that:

 \begin{align*}
u_1(1,3,3)&=p_1\tau_1\\ 
u_1(2,3,3)&=p_1\tau_1+(1-p_1)p_2(\tau_1-\text{dist}(2,1))+(1-p_1)(1-p_2)p_3(\tau_1-\text{dist}(3,1))\\ u_1(3,3,3)&=p_1\tau_1+(1-p_1)p_3(\tau_1-\text{dist}(3,1)).
\end{align*}
Since $\text{dist}(2,1)<\text{dist}(3,1)$ we also have that 

 \begin{align*}
u_1(2,3,3)&>p_1\tau_1+(1-p_1)p_2(\tau_1-\text{dist}(3,1))+(1-p_1)(1-p_2)p_3(\tau_1-\text{dist}(3,1))
\\ 
&=p_1\tau_1+(1-p_1)(\tau_1-\text{dist}(3,1))p_3=
u_1(3,3,3).
\end{align*}
Hence, we can assume that $d(1) \neq 3$. Then, we can determine the $\argmax_{i\in \{1,2\}}\{u_1(i,3,3)\},$ and set $d(i)$ to be equal to that value.
Note that whatever the determined delegation choice for voter $1$ is, it does not affect the utility of voter $2$, so, for now, voter $2$ does not have an incentive to deviate.

We continue with fixing a delegation for voter $3$ and we now have two cases to consider, i.e., $d(1)=1$ and $d(1)=2$. 
\begin{itemize}
    \item We first examine the case where $d(1)=1$. Then, by mutual acceptance, $u_3(d(1),d(2),1)<u_3(d(1),d(2),3),$ or otherwise voter $1$ would prefer a delegation to voter $3$ than $d(1)=1$. Similarly, $u_3(d(1),d(2),2)>u_3(d(1),d(2),3),$ or otherwise voter $2$ would also not prefer a delegation to voter $3$. Consequently, we always fix $d(3)=2$ in this case. It is easy to check that no voter has an incentive to deviate, or otherwise voters would not have fixed the corresponding choices in the previous steps of the procedure.
    \item Suppose now that $d(1)=2.$ There are two possible options for voter $3$ to consider: delegating to voter $1$ and hence forming a 3-cycle, or delegating to voter $2$ and hence forming a 2-cycle. We select for voter $3$ the option among the two that maximizes their own utility, given that $d(2)=3$ and $d(1)=2$. We will show that in both cases neither voter $1$ nor voter $2$ have an incentive to deviate. More precisely, say first that $\de$ induces a 3-cycle, and voter $1$ wants to deviate. This can only be when $\tau_1<dist(1,j)$ for some $j\in \{2,3\}$. This cannot hold for $j=2$ or otherwise voter $1$ wouldn't have delegated to $2$ beforehand and it cannot also hold for $j=3$ or otherwise voter $3$ wouldn't have delegated to $1$. The argument for rejecting a potential deviation of voter $2$ is similar.   
    Now say that $\de$ induces a 2-cycle between voters $2$ and $3$ and $d(1)=2$. The fact that at a previous step voter $2$ selected to delegate to $3$ means that a deviation of voter $2$ will not be profitable for them. It remains to show that voter $1$ doesn't have a profitable deviation either. To that end, we will prove that $u_1(\de_{-1},2)>u_1(\de_{-1},3).$ This is true due to the following equivalent expressions: 

      (i) $(\tau_1-\text{dist}(1,2))p_2+(\tau_1-\text{dist}(1,3))p_3(1-p_2)    > 
                (\tau_1-\text{dist}(1,3))p_3+(\tau_1-\text{dist}(1,2))p_2(1-p_3)$,\\ (ii) $
                (\tau_1-\text{dist}(1,2))p_2 p_3> (\tau_1-\text{dist}(1,3))p_3 p_2$, and\\ (iii)     
                $\text{dist}(1,2)< \text{dist}(1,3)$.
\end{itemize}
Thus, the sequence in which voters are examined, along with the criteria that led to the specification of their delegations, resulted in a Nash equilibrium.

It remains to consider the case where $\max_{i\in \{1,3\}}\{p_i(\tau_2-\text{dist}(2,i))\}\leq 0$. Equivalently, $(\tau_2-\text{dist}(2,i))\leq 0,$ for $i\in \{1,3\}$. If this is the case, we set $d(i)=i$ for every $i\in \{1,2,3\}$. Then, each voter $i$ experiences a utility of $p_i\tau_i$. Consider a possible deviation for voter $2$, say by delegating to some voter $i\in \{1,3\}$. Then their new utility becomes $p_2\tau_2+(1-p_2)p_i(\tau_2-\text{dist}(2,i))\leq p_2\tau_2$. Therefore, such a deviation is not profitable for voter $2$. 
The same argument holds for all the remaining possible deviations, since, by mutual acceptance, if voter $1\notin \mathcal{A}_2(\x,\Btau)$ and $3\notin \mathcal{A}_2(\x,\Btau)$, then also 
$2\notin \mathcal{A}_1(\x,\Btau)$ and $2\notin \mathcal{A}_3(\x,\Btau),$ so both voters $1$ and $3$ would prefer to loop rather than delegating to voter $2$, let alone to voter $3$ and $1$ respectively.
\hfill $\qed$
\end{enumerate}

 \subsection*{Proof of \cref{thm:no-NE}}

 \begin{enumerate}
 \item[(i)] Let us first show that there are instances without a NE even with only three voters present. Take an instance with three voters, $A, B$, and $C,$ and say that $\x = (0.1, 0.9, 0.95)$, $\p = (0.02, 0.02, 0.1)$, and $\Btau = (0.9, 0.1, 0.9)$. Assume that there exists a NE profile $\d$ in that instance. First, notice that $B$ does not choose to delegate to $A$ in $\d$, as $\text{dist}(A,B) > \tau_B$. Second, we observe that $A$ and $C$ do not self-delegate, as $\text{dist}(A,B) < \tau_A$ and $\text{dist}(C,B) < \tau_C$. \Cref{fig:NoNEnfg3ag} depicts the utilities of the voters in all the remaining strategy profiles. It is routine to check that there is no NE in this instance.

\begin{figure}[t]
\centering
\scalebox{.7}{\nfgameN{$d(A)=B$ $d(A)\rightarrow C$ $d(B)=C$ $d(B)=B$ 0.0247,\,\,\,-0.0054,\,\,\,0.1082 0.0229,\,\,\,-0.0054,\,\,\,0.0909 0.0199,\,\,\,0.002,\,\,\,0.1058 0.0229,\,\,\,0.002,\,\,\,0.0909 }}
~
\hspace{1cm}
\scalebox{.7}{\nfgameN{$d(A)=B$ $d(A)=C$ $d(B)=C$ $d(B)=B$ 0.0247,\,\,\,0.0069,\,\,\,0.1053 0.0246,\,\,\,0.0069,\,\,\,0.1053 0.0199,\,\,\,0.002,\,\,\,0.1053 0.0246,\,\,\,0.002,\,\,\,0.1053 }}
\caption{Normal form representation of the game induced by the instance described in the proof of \cref{thm:no-NE}. The table on the left corresponds to $C$ delegating to $A$, and the one on the right to $C$ delegating to $B$. Rows correspond to the choices of voter $A$ and columns to the choices of voter $B$. The numbers in each cell show the utility of $A$, $B$, and $C$ respectively, for the delegation profile under examination.
}\label{fig:NoNEnfg3ag}
\end{figure}
     \item [(ii)]
Consider the following instance $\inst$: 

 \begin{align*}
    p&=(0.05, 0.49, 0.1, 1.0),\\ x&=(0.0, 0.05, 0.15, 0.3),\\ \tau&=0.2.
\end{align*}

We name the voters of the instance as $\N=\{0,1,2,3\}$.
We first observe that it is sufficient to focus on delegation profiles where voter $3$ loops. This is because if there is a Nash equilibrium where voter $3$ doesn't loop, there is another delegation profile where they loop and that it is a Nash equilibrium as well.
All possible delegation profiles $\de$ in which $d(3)=3$ are depicted in (the first column of) \cref{tab:thm5}. For each such $\de$ we identify a voter $i\in \N$ for which there exists a profitable deviation, proving that no Nash equilibrium exists in $\inst$. \hfill \qed 
\end{enumerate}

\begin{table*}[t]
\resizebox{1\columnwidth}{!}{
\centering
\begin{tabular}{ccccc}
\toprule
\textbf{Delegation Profile} & \textbf{Deviating Voter} & \textbf{Current Delegation} & \textbf{Profitable Deviation} & \textbf{Utilities for Deviating voter} \\ \midrule
$[0, 0, 0, 3]$ &   0 & 0 & 1 & [0.0100, 0.0798, 0.0148, -0.0850] \\ 
$[0, 0, 1, 3]$ &   0 & 0 & 1 & [0.0100, 0.0798, 0.0776, -0.0850] \\ 
$[0, 0, 2, 3]$ &   0 & 0 & 1 & [0.0100, 0.0798, 0.0148, -0.0850] \\ 
$[0, 0, 3, 3]$ &   0 & 0 & 1 & [0.0100, 0.0798, -0.0707, -0.0850] \\ 
$[0, 1, 0, 3]$ &   0 & 0 & 1 & [0.0100, 0.0798, 0.0148, -0.0850] \\ 
$[0, 1, 1, 3]$ &   0 & 0 & 1 & [0.0100, 0.0798, 0.0776, -0.0850] \\ 
$[0, 1, 2, 3]$ &   0 & 0 & 1 & [0.0100, 0.0798, 0.0148, -0.0850] \\ 
$[0, 1, 3, 3]$ &   0 & 0 & 1 & [0.0100, 0.0798, -0.0707, -0.0850] \\ 
$[0, 2, 0, 3]$ &   0 & 0 & 1 & [0.0100, 0.0822, 0.0148, -0.0850] \\ 
$[0, 2, 1, 3]$ &   0 & 0 & 1 & [0.0100, 0.0822, 0.0776, -0.0850] \\ 
$[0, 2, 2, 3]$ &   0 & 0 & 1 & [0.0100, 0.0822, 0.0148, -0.0850] \\ 
$[0, 2, 3, 3]$ &   0 & 0 & 1 & [0.0100, 0.0386, -0.0707, -0.0850] \\ 
$[0, 3, 0, 3]$ &   0 & 0 & 1 & [0.0100, 0.0314, 0.0148, -0.0850] \\ 
$[0, 3, 1, 3]$ &   0 & 0 & 2 & [0.0100, 0.0314, 0.0340, -0.0850] \\ 
$[0, 3, 2, 3]$ &   0 & 0 & 1 & [0.0100, 0.0314, 0.0148, -0.0850] \\ 
$[0, 3, 3, 3]$ &   0 & 0 & 1 & [0.0100, 0.0314, -0.0707, -0.0850] \\ 
$[1, 0, 0, 3]$ &   1 & 0 & 2 & [0.1018, 0.0980, 0.1065, 0.0725] \\ 
$[1, 0, 1, 3]$ &   1 & 0 & 2 & [0.1018, 0.0980, 0.1031, 0.0725] \\ 
$[1, 0, 2, 3]$ &   1 & 0 & 2 & [0.1018, 0.0980, 0.1031, 0.0725] \\ 
$[1, 0, 3, 3]$ &   2 & 3 & 1 & [0.0641, 0.0652, 0.0200, 0.0650] \\ 
$[1, 1, 0, 3]$ &   1 & 1 & 2 & [0.1018, 0.0980, 0.1065, 0.0725] \\ 
$[1, 1, 1, 3]$ &   1 & 1 & 2 & [0.1018, 0.0980, 0.1031, 0.0725] \\ 
$[1, 1, 2, 3]$ &   1 & 1 & 2 & [0.1018, 0.0980, 0.1031, 0.0725] \\ 
$[1, 1, 3, 3]$ &   1 & 1 & 0 & [0.1018, 0.0980, 0.0802, 0.0725] \\ 
$[1, 2, 0, 3]$ &   2 & 0 & 3 & [0.0641, 0.0641, 0.0200, 0.0650] \\ 
$[1, 2, 1, 3]$ &   2 & 1 & 3 & [0.0641, 0.0641, 0.0200, 0.0650] \\ 
$[1, 2, 2, 3]$ &   2 & 2 & 3 & [0.0641, 0.0641, 0.0200, 0.0650] \\ 
$[1, 2, 3, 3]$ &   1 & 2 & 0 & [0.1018, 0.0980, 0.0802, 0.0725] \\ 
$[1, 3, 0, 3]$ &   1 & 3 & 2 & [0.1018, 0.0980, 0.1065, 0.0725] \\ 
$[1, 3, 1, 3]$ &   0 & 1 & 2 & [0.0100, 0.0314, 0.0340, -0.0850] \\ 
$[1, 3, 2, 3]$ &   1 & 3 & 2 & [0.1018, 0.0980, 0.1031, 0.0725] \\ 
$[1, 3, 3, 3]$ &   1 & 3 & 0 & [0.1018, 0.0980, 0.0802, 0.0725] \\ 
$[2, 0, 0, 3]$ &   0 & 2 & 1 & [0.0100, 0.0798, 0.0148, -0.0850] \\ 
$[2, 0, 1, 3]$ &   0 & 2 & 1 & [0.0100, 0.0798, 0.0776, -0.0850] \\ 
$[2, 0, 2, 3]$ &   0 & 2 & 1 & [0.0100, 0.0798, 0.0148, -0.0850] \\ 
$[2, 0, 3, 3]$ &   0 & 2 & 1 & [0.0100, 0.0798, -0.0707, -0.0850] \\ \bottomrule
\end{tabular}}
\caption{For each delegation profile $\de$ satisfying $d(3)=3$ (column 1), we identify a voter $i\in N$ (column~2) who can improve their expected utility by unilaterally changing their delegation from $d(i)$ (column 3) to $d(i)'$ (column~4). Column~5 shows the expected utility voter $i$ would obtain by delegating to each voter in $\N$, while keeping the other voters' delegations unchanged.}
\label{tab:thm5}
\end{table*}

%
 
 \subsection*{Proof of \cref{thm: NE-LR}}

We focus on rightists; the case for leftists is analogous. The proof proceeds by a greedy method. We consider the voters in order of decreasing position, to be denoted as $v_1, v_2, \cdots, v_n$. Voter $v_n$, which corresponds to the voter at position $x_{\max}=\argmax_{i\in \N}\{x_i\}$, will self-loop in an equilibrium, simply due to the fact that, as long as $p_n\neq1,$ any delegation will result in being represented (with some probability) by someone in their left, which gives a negative utility to a rightist voter. If $p_n=1,$ self-looping, trivially, doesn't admit a profitable deviation. 

Consider now voter $v_{n-1}$. Being a rightist, $v_{n-1}$ can only delegate to themselves or to $v_n$ to obtain a non-negative utility. The decision of this voter is not affected by delegations to them, or by edges to voters on their left. Therefore, the delegation choices of voters $v_1, \cdots, v_{n-2}$, which will be fixed in a later step of the procedure, cannot influence the decision of $v_{n-1}$. Thus, for $v_{n-1},$ we only need to check whether delegating to $v_n$ or to themselves produces more expected utility (knowing that $v_n$ loops) and fix the corresponding choice. 

We proceed similarly for the remaining voters in order. Notice that the decision of $v_{n-1}$ impacts only the decisions of voters positioned before $v_{n-1}$, who will be examined next. In the specified order, when it is time to consider a voter with a lower position, they will choose the best delegation option available to them, taking into account the already-fixed delegations of voters positioned on their right and the fact that they cannot delegate to anyone positioned on their left. Once such a delegation is chosen, it becomes fixed, as changes in the delegations of voters positioned on their left will not affect their expected utility and make them interested in deviating. \hfill $\qed$

 \subsection*{Proof of \cref{thm: NE-proxy}}

We will present a procedure that identifies a Nash equilibrium $\de,$ in a given instance of the proxy voting setting. Initially we call all voters as \emph{unclassified}. Additionally, we will refer to the following sets of voters: \emph{represented} and \emph{representatives}, initially both empty.

For each unclassified voter $i$ in an arbitrary order, we determine the voter $j$ maximizing the quantity $p_j(\tau_i-\text{dist}(i,j))$ among all $j$ such that $j$ doesn't yet belong to the set of represented voters ($i$ included). Note that if $j$ is represented, a delegation from $i$ to $j$ would create a chain of $3$ voters, hence such a delegation is infeasible. Then, we fix $d(i)=j$ and we remove both $i$ and $j$ from the set of unclassified voters, labeling $i$ as a represented and $j$ as a representative. Subsequently, we repeat for another unclassified voter. When all unclassified voters have been examined, we move to the second phase of the algorithm. 

At the end of the described first phase, all voters are either represented or representatives but representatives do not have an outgoing edge yet in $G_{\d}$ (or they have a self-loop). For those without a loop, it remains to determine their delegation. Then, we consider each such a voter in arbitrary order. 
Say we examine voter $i$, the only choices for $i$ are to delegate to one of the voters in $\{k \in \N \mid d(k)=i\},$ i.e. those who already delegated to $i$ or to self-loop, due to the constraint of the proxy voting setting. Among those options, $i$ picks the delegation that maximizes their utility, which again corresponds to the voter $j\in \{k \in \N \mid d(k)=i\}\cup \{i\},$ maximizing the quantity $p_j(\tau_i-\text{dist}(i,j))$. 

To prove that the constructed delegation forms indeed a Nash equilibrium, we consider a voter $i$ and we assume that they prefer to deviate, towards a contradiction. First, we examine a voter $i$ that self-delegates. By construction, there are no feasible choices for $i$ that would improve their utility. Assume now that $i$ is a voter that is in the set of represented voters and didn't self-loop. Then, $i$ is delegating to the one giving them the maximal utility, at the moment we considered $i$ in the first round. It holds that $i$ cannot delegate to anyone that delegates further, so every other voter in the set of represented voters is not a feasible choice. Among the representatives, $i$ chose the best for themselves option. Suppose now that $i$ is in the set of representatives and didn't self-loop. Then, simply, no deviation to a voter that doesn't delegate to $i$ is feasible for $i$, and among those, $i$ has been assigned to their preferred one.
\hfill $\qed$


 

 \subsection*{Proof of \cref{obs:NEcycles}}

Consider a weakly connected component of a Nash equilibrium $\de$ of $\inst$ that has at least $2$ vertices, or, in other words, at least one (directed) edge. Towards a contradiction, suppose that it does not have a directed cycle. Then, its undirected variant must be a tree, meaning it has a sink vertex that self-delegates. Let this vertex correspond to a voter $i$, and let $j$ be a voter such that $d(j) = i$. The utility of $i$ under $\de$ equals $\tau_i p_i$. 

If voter $i$ were to delegate back to voter $j$, then $i$ would obtain a utility of $\tau_i p_i + (1-p_i)(\tau_i - \text{dist}(i,j))p_j$. Since the delegation from $i$ to $j$ was not selected in the equilibrium $\de$ (otherwise $i$ would not have been a sink), it must hold that  

\begin{equation*} 
\tau_i p_i + (1-p_i)(\tau_i - \text{dist}(i,j))p_j \leq \tau_i p_i
\Leftrightarrow (1-p_i)(\tau_i - \text{dist}(i,j))p_j \leq 0.
\end{equation*}  

Therefore, it must hold that $p_i = 1$, or $\tau_i \leq \text{dist}(i,j)$, or $p_j = 0$. The first and third cases do not hold by assumption. Consequently, it is necessarily true that  

\begin{equation}  
\label{eq:contr}  
\tau_i - \text{dist}(i,j) \leq 0.  
\end{equation}  

Since $j$ delegates to $i$ at $\de$, the utility of $j$ equals $\tau_j p_j + (1-p_j)(\tau_j - \text{dist}(i,j))p_i$, and this value must be no less than $\tau_j p_j$, or $j$ would prefer to self-loop. Therefore,  

\begin{equation*}
    (1-p_j)(\tau_j - \text{dist}(i,j))p_i \geq 0 \Rightarrow \tau_j - \text{dist}(i,j) \geq 0.
\end{equation*}  

By assumption, $\tau_j - \text{dist}(i,j) \neq 0$, and hence, $\tau_j - \text{dist}(i,j) > 0$. Consequently, $i \in \mathcal{A}_j(\x,\Btau)$, and by the assumption of mutual acceptance between voters, $j \in \mathcal{A}_i(\x,\Btau)$. Therefore, $\tau_i - \text{dist}(i,j) > 0$, which contradicts \cref{eq:contr}.  
Thus, the considered weakly connected component must have at least one directed cycle. 

We will now prove that no more than a single cycle may appear in a weakly connected component. Consider such a component with at least one cycle and let it contain $k$ vertices. Since each vertex has out-degree 1, the component also has $k$ edges. Removing now one edge from the cycle leaves the component connected, with $k$ vertices and $k-1$ edges. This structure is necessarily a tree, hence, the original component can be viewed as a tree plus one additional edge, and adding a single edge to a tree creates no more than a single cycle, proving the claim.
\hfill $\qed$


 \subsection*{Proof of \cref{prop:degenerate}}

We will show that there exists an instance $\inst$ with a Nash equilibrium $\de^*$ such that there is a weakly connected component of at least two vertices in $G_{\de^*}$ that does not have a cycle. The instance $\inst$ will have at least one of the following properties: (1) for some voter $i\in \N$, we have that $p_i=1$, (2) for some voter $j\in \N$, we have that $p_j=0$, (3) there is a pair of voters $(i,j) \in \N\times \N$ such that $i\in \mathcal{A}_j(\x,\Btau)$ but $j\notin \mathcal{A}_i(\x,\Btau)$.

Consider an instance $\inst$ with only two voters, namely $i$ and $j$, where $i \in \mathcal{A}_j(\x,\Btau)$. We call $\de^*$ the profile in which $i$ self-delegates and $j$ delegates to $i$. Observe that no cycle appears in the connected component of size $2$ of $G_{\de^*}$. Note that $j$ receives a utility under $\de^*$ that is at least as large as under $(\d^*_{-j}, j)$, if any of the conditions (1)-(3) hold. It remains to prove that for $i$ it is not better to delegate to $j$ instead. 
It is first immediate that if $p_i=1$, then $i$ receives the same expected utility by self-delegating as by delegating to $j$. Also, if $p_j=0$, then delegating to $j$ provides the same utility to $i$ as self-delegation. Then, if $j\notin \mathcal{A}_i(\x,\Btau)$, $i$  receives strictly lower utility in $(\d^*_{-i}, j)$ than in $\d^*$. 
Therefore, in all of the examined cases $\d^*$ is a Nash equilibrium.
\hfill $\qed$

 \subsection*{Proof of \cref{obs:NEcycles_degenerate}}

  Suppose that $\de$ is a Nash equilibrium and that $G_{\de}$ contains a weakly connected component $W$ of more than one vertex that does not have a directed cycle. Then there exists a voter $i$ in $W$ that is self-delegating under $\de$ and a voter $j$ such that $d(j) = i.$ 
 If $p_j\neq 0$ and $p_i\neq 1,$ the proof of \cref{obs:NEcycles} works. So, we first assume that $p_i=1$. We claim that there is another NE profile $\de'$, where $d(j)'=i$ and $d(i)'=j$. The fact that $p_i=1$ means that voter $i$ as well as every voter delegating directly or indirectly to voter $i$ under $\de$ is not affected by whether $i$ self-delegates or not. The same holds trivially for the rest of the voters, so $\de'=(\de_{-i},j)$ is a Nash equilibrium as well. 
The proof is similar for the case where for the voter $j$ it holds that $p_j=0$. Specifically, we claim again that there is another NE $\de'$ such that $d'_j=i$ and $d'_i=j$. The fact that $p_j=0$ means that the change of delegation for voter $i$ from self-looping to $j$ will not affect the utility of any voter as $j$ will never cast a ballot. Therefore, the expected utility of every voter is the same under both delegation profiles.  
\hfill $\qed$

 \subsection*{Proof of \cref{obs:entrypoints}}

Take a weakly connected component $W$ of a graph $G_{\de}$ where $\de$ is a Nash equilibrium of an instance $\inst$ satisfying the conditions of the statement. 
For a voter $v \in \mathcal{L}(W) \cup \mathcal{R}(W)$ and a voter $w \in \mathcal{C}(W)$ we denote as $\mathcal{Q}^v_w(w')$ the probability that a vertex $w' \in \mathcal{C}(W)$ corresponds to the ultimate delegate of $v,$ (to be called $\mathcal{P}^v_w(w')$) times $(1-p_v)$ in the case that $v$ delegates to $w$.
We note that $\mathcal{P}^v_w$ is actually the same for every voter $v \in \mathcal{L}(W) \cup \mathcal{R}(W)$ and hence we will drop the superscript. 

The following lemma is a direct consequence of the definition of expected utility together with the fact that no voter in $\mathcal{C}(W)$ delegated to a voter outside of $\mathcal{C}(W)$ according to $\de$.

\begin{lemma}\label{lem:CycleSimpleUti}
For a voter $v \in \mathcal{L}(W) \cup \mathcal{R}(W)$ that delegates to a voter $w \in \mathcal{C}(W)$, the expected utility of $v$ is $p_v \tau + \sum_{i \in \mathcal{C}(W)} \mathcal{Q}_{w}^v(i) \cdot (\tau - \text{dist}(v,i))= p_v \tau + (1-p_v)\sum_{i \in \mathcal{C}(W)} \mathcal{P}_{w}(i) \cdot (\tau - \text{dist}(v,i))$.
\end{lemma}

Suppose, without loss of generality, that there are two distinct voters $v, v'\in \mathcal{L}(W)$ that are delegating to a vertex in $\mathcal{C}(W)$, under $\de$. Assume further that $x_v < x_{v'}$. First, 
notice that since both $v$ and $v'$ are positioned to the left of all voters in $\mathcal{C}(W)$, from \Cref{lem:CycleSimpleUti}, we have that the set of voters in $\mathcal{C}(W)$ maximizing the expected utility if delegated to is the same for $v$ and $v'$ and, since $\d$ is a NE, both of them delegate to such a voter, say $w$.

We will show now that $\d$ is not a NE by showing that $v$ would benefit from switching their delegation to $v'$ instead of $w$. In that case, the utility of $v$ would be

\begin{equation*}
\tau p_v + (1- p_v)\left(p_{v'} (\tau - \text{dist}(v, v')) + \sum_{i \in \mathcal{C}(W)} \mathcal{P}_{w}(i)(1- p_{v'}) (\tau - \text{dist}(v,i))\right).
\end{equation*}
Now, recall that the expected utility of $v$ under $\d$ amounts to 
$
 \tau p_v + (1-p_v)\sum_{i \in \mathcal{C}(W)} \mathcal{P}_{w}(i) (\tau - \text{dist}(v,i)).
$
It suffices to show 

 \begin{align*}
p_{v'} (\tau - \text{dist}(v, v')) + \sum_{i \in \mathcal{C}(W)} \mathcal{P}_{w}(i)(1- p_{v'}) (\tau - \text{dist}(v,i)) 
&> 
\sum_{i \in \mathcal{C}(W)} \mathcal{P}_{w}(i) (\tau - \text{dist}(v,i)) \Leftrightarrow\\
(\tau - \text{dist}(v, v'))  
&>
\sum_{i \in \mathcal{C}(W)} \mathcal{P}_{w}(i) (\tau - \text{dist}(v,i)),
\end{align*}
which is true because

\begin{align*}
    \sum_{i \in \mathcal{C}(W)} \mathcal{P}_{w}(i) (\tau - \text{dist}(v,i)) &<   \sum_{i \in \mathcal{C}(W)} \mathcal{P}_{w}(i) (\tau - \text{dist}(v,v')) 
    = \\ (\tau - \text{dist}(v,v')) \sum_{i \in \mathcal{C}(W)} \mathcal{P}_{w}(i) &\leq \tau - \text{dist}(v,v'),
\end{align*}
where the last inequality holds as $\sum_{i \in \mathcal{C}(W)}\mathcal{P}_{w}(i) \leq 1$.

Finally, because $W$ is a connected component containing exactly one cycle (by \cref{obs:NEcycles}), it directly follows that $\mathcal{L}(W)$ and $\mathcal{R}(W)$ induce directed (in-)trees, which are directed towards $\mathcal{C}(W)$.
\hfill $\qed$

\subsection*{Proof of \cref{SCno-cycle}}

Consider the following symmetric instance of three voters, namely $A$, $B$, and $C$, with $\tau=0.4$:

 \begin{align*}
    x=(0.12, 0.5, 0.88), \qquad p=(0.1, 0.9, 0.1).
\end{align*}

Notice first that in a profile that maximizes total utility $A$ and $C$ do not delegate to each other, as $\text{dist}(A,C) > 0.4$, and hence changing a delegation of one of them to $B$ would strictly improve the total utility.
The social welfare of the rest possible delegation profiles is shown in \cref{tab:delegation-socialwelfare}.
\begin{table}[h]
\centering
\begin{tabular}{@{}cc@{}}
\toprule
\textbf{Delegation Profile} & \textbf{Social Welfare} \\ \midrule
$[A, A, B]$ & 0.4532 \\
$[A, A, C]$ & 0.4402 \\
$[A, B, B]$ & 0.4562 \\
$[A, B, C]$ & 0.4400 \\
$[A, C, B]$ & 0.4564 \\
$[A, C, C]$ & 0.4402 \\
$[B, A, B]$ & 0.4694 \\
$[B, A, C]$ & 0.4564 \\
$[B, B, B]$ & \textbf{0.4724} \\
$[B, B, C]$ & 0.4562 \\
$[B, C, B]$ & 0.4694 \\
$[B, C, C]$ & 0.4532 \\
\bottomrule
\end{tabular}
\caption{Social Welfare achieved by profiles in which $A$ doesn't delegate to $C$ and vice versa, in the instance examined in the proof of \cref{SCno-cycle}. The maximizing value of social welfare appears in bold.}
\label{tab:delegation-socialwelfare}
\end{table}

\noindent It follows that $\d=[B,B,B]$ maximizes total utility; its delegation graph appears below.

    \begin{figure}[h!]
\centering
\scalebox{.6}{\begin{tikzpicture}[node distance=0.2cm,semithick]

\node at  (0,0.5) {$0$};
\node at  (1.2,0.5) {$0.12$};
\node at  (5,0.5) {$0.5$};
\node at  (8.8,0.5) {$0.88$};
\node at  (10,0.5) {$1$};

\node (A) at  (1.2,0.7-1) {$A$};
\node (B) at  (5,0.7-1) {$B$};
\node (C) at  (8.8,0.7-1) {$C$};

\path[thick] (0,0.2) edge (10,0.2)
             (0,0.25) edge (0,0.1)
             (1.2,0.2) edge (1.2,0.15)
             (5,0.2) edge (5,0.15)
             (8.8,0.2) edge (8.8,0.15)
             (10,0.25) edge (10,0.1);

\draw[thick, -latex] (A) edge[bend right] (B); 
\draw[thick, -latex] (C) edge[bend left] (B); 
\draw[thick, -latex] (B) edge[loop below] (B); 
\end{tikzpicture}}%
\end{figure}

\hfill \qed

\subsection*{Proof of \cref{thm: poa}}

We describe a family of instances $\inst_n$, parameterized by the number of
voters $n$, for which the (unique) Nash equilibrium $\de$ satisfies
$SW(\de) \to 0$ as $n \to \infty$, while the social-welfare-optimal profile
satisfies $SW(\de_{\texttt{SW}}) = \Omega(1)$; hence
$\textit{PoA}(\inst_n) \to \infty$.

Fix a constant $\lambda \in (0,1)$. Consider an instance $\inst_n$ of $n$
voters, where for $i \in \{2,\dots,n\}$, voter $i$ is positioned at a distance
of $1/n^3$ after voter $i-1$, and voter $1$ is positioned at $x_1=0$. For
voter $v_1$ it holds that $\tau_1 = \lambda$, whereas for the rest of the
voters it holds that $\tau_i = 0$. All voters vote with
probability $1/n$. Note that all voters lie in $[0,\frac{n-1}{n^3}]
\subseteq [0,\frac1{n^2}]$, so $\dist(i,j) < \frac1{n^2}$ for every pair
$i,j$; in particular $\mathcal{A}_1(\x,\Btau) = V$ for $n$ large enough, while
$\mathcal{A}_i(\x,\Btau) = \{i\}$ for every $i \ge 2$. For an illustration see \Cref{fig:thm10}.

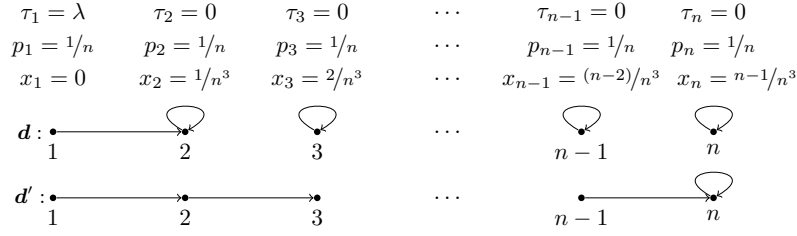
\begin{figure}[h!]
\centering
\resizebox{0.6\columnwidth}{!}{
\centering

\begin{tikzpicture}

\def\spacing{2}

\node[circle, fill=black, inner sep=1pt, label=below:$1$] (v1) at (0, 0) {};
\node at (0, 0.8) {$x_1 = 0$};
\node at (0, 1.3) {$p_1 = \nicefrac{1}{n}$};
\node at (0, 1.8) {$\tau_1 = \lambda$};

\node[circle, fill=black, inner sep=1pt, label=below:$2$] (v2) at (\spacing, 0) {};
\node at (\spacing, 0.8) {$x_2 = \nicefrac{1}{n^3}$};
\node at (\spacing, 1.3) {$p_2 = \nicefrac{1}{n}$};
\node at (\spacing, 1.8) {$\tau_2 = 0$};

\node[circle, fill=black, inner sep=1pt, label=below:$3$] (v3) at (2*\spacing, 0) {};
\node at (2*\spacing, 0.8) {$x_3 = \nicefrac{2}{n^3}$};
\node at (2*\spacing, 1.3) {$p_3 = \nicefrac{1}{n}$};
\node at (2*\spacing, 1.8) {$\tau_3 = 0$};

\node at (3*\spacing, 0.8) {$\dots$};
\node at (3*\spacing, 1.3) {$\dots$};
\node at (3*\spacing, -0) {$\dots$};
\node at (3*\spacing, 1.8) {$\dots$};

\node[circle, fill=black, inner sep=1pt, label=below:$1$] (v11) at (0, -1) {};

\node[circle, fill=black, inner sep=1pt, label=below:$2$] (v21) at (\spacing, -1) {};

\node[circle, fill=black, inner sep=1pt, label=below:$3$] (v31) at (2*\spacing, -1) {};

\node[circle, fill=black, inner sep=1pt, label=below:$n-1$] (vn2) at (4*\spacing, -1) {};
\node[circle, fill=black, inner sep=1pt, label=below:${n}$] (vn11) at (5*\spacing, -1) {};

\node at (3*\spacing, -1) {$\dots$};

\node[circle, fill=black, inner sep=1pt, label=below:$n-1$] (vn) at (4*\spacing, 0) {};
\node at (4*\spacing, .8) {$x_{n-1} = \nicefrac{(n-2)}{n^3}$};
\node at (4*\spacing, 1.3) {$p_{n-1} = \nicefrac{1}{n}$};
\node at (4*\spacing, 1.8) {$\tau_{n-1} = 0$};

\node[circle, fill=black, inner sep=1pt, label=below:${n}$] (vn1) at (5*\spacing, 0) {};
\node (v5)at (5.18*\spacing, .8) {$x_{n} = \nicefrac{n-1}{n^3}$};
\node (v6) at (5*\spacing, 1.3) {$p_{n} = \nicefrac{1}{n}$};
\node at (5*\spacing, 1.8) {$\tau_{n} = 0$};





\draw[->] (v11) to (v21) ;
\draw[->] (v21) to (v31) ;
\draw[->] (vn2) to (vn11) ;
\draw[->] (vn11) edge [loop below, out=150, in=30,  looseness=25]  (vn11);



\draw[->] (v1) to (v2);
\draw[->] (v2) edge [loop below, out=150, in=30,  looseness=25, yshift=-85cm]  (v2);
\draw[->] (v3) edge [loop below, out=150, in=30,  looseness=25, yshift=-85cm]  (v3);
\draw[->] (vn) edge [loop below, out=150, in=30,  looseness=25, yshift=-85cm]  (vn);
\draw[->] (vn1) edge [loop below, out=150, in=30,  looseness=25, yshift=-85cm]  (vn1);








\node[circle,  inner sep=1pt, 
left of =v1, xshift=.65cm, yshift=-0cm](d){\small $\de:$} ;
\node[circle,  inner sep=1pt, 
below of =d, xshift=0cm, yshift=.05cm](d1){\small $\de':$} ;


\end{tikzpicture}}
\caption{Illustration of the instance used in the proof of \Cref{thm: poa}.
}
\label{fig:thm10}
\end{figure}

The profile $\de$ where voter $1$ delegates to voter $2$ and all others delegate to
themselves is a Nash equilibrium. Indeed, for any voter $i \ge 2$ and any
fixed delegations of the others, the term of $u_i$ contributed by $i$
casting their own ballot equals $p_i \tau_i = 0$ regardless of $i$'s
delegation choice. If $i$ delegates to
some other casting voter, say voter $y$ with $\dist(i,y) > 0$, this is contributing a strictly
negative term with positive probability. Hence self-delegation is a strictly
dominant strategy for every $i \ge 2$, regardless of the delegations of the other
voters. Given that voters $2,\dots,n$ self-loop, voter $1$ prefers
delegating to $2$ rather than to any other voter, since all voters vote
with the same probability and voter $2$ is the nearest to voter $1$ and belongs in their acceptability set.
We now compute the total utility gained in this profile. It holds that
$u_i(\de) = 0$ for each $i \in \{2,\dots,n\}$. Therefore:

\[
  SW(\de) = u_1(\de)
    = \frac{\lambda}{n} + \Big(1-\frac1n\Big)\frac1n\Big(\lambda - \frac1{n^3}\Big)
    \;\le\; \frac{2\lambda}{n} + \frac{1}{n^5}
    \;\xrightarrow{\,n \to \infty\,}\; 0
\]

Next, consider the profile $\de'$ such that voter $i$ delegates to voter $i+1$ for
$i \in \{1,\dots,n-1\}$, and voter $n$ to themselves. Clearly, this is not an
equilibrium, since there exists a voter that has an incentive to deviate and
self-loop. In order to show that PoA is unbounded, it suffices to show that
the total utility obtained from this profile is bounded away from $0$.
We start by computing the utility of voter $1$. Writing $r = 1 - \frac1n$ and
using the sum of geometric series together with
the fact that the distance between voters $1$ and $\ell+1$ for $\ell = 0,1,2\dots,n-1$ is equal to $\ell/n^3$ we have that

\[
  u_1(\de') = \sum_{\ell=0}^{n-1}\Big(\lambda - \frac{\ell}{n^3}\Big)
    \frac1n\, r^\ell
    = \lambda\big(1 - r^n\big) - \frac1{n^4}\sum_{\ell=0}^{n-1}\ell\, r^\ell .
\]
Since $\sum_{\ell=0}^{n-1}\ell\, r^\ell \le \sum_{\ell = 0}^{\infty}\ell\, r^\ell =
\frac{r}{(1-r)^2} < n^2$, the second term is bounded by $1/n^2 \to 0$; and
$r^n = (1-\frac1n)^n \to e^{-1}$. Hence

\[
  u_1(\de') \xrightarrow{\,n \to \infty\,} \lambda\big(1 - e^{-1}\big).
\]
We next turn to the utility of the remaining voters. 
For voter $n$ it holds that $u_n(\de')=0.$
For $i \in \{2,3,\dots,n-1\}$, $\tau_i =
0$ gives $u_i(\de') = -\sum_{y=i+1}^{n} \dist(i,y)\Pr[y],$
where $\Pr[y]$ denotes the probability that $y$ votes on behalf of $i,$ that is, no one between \(i\) and \(y\) votes while \(y\) does.
It also holds that 
$-\sum_{y=i+1}^{n} \dist(i,y)\Pr[y]
\ge -\frac1{n^2}$, since every
voter $y$ reachable from $i$ satisfies $\dist(i,y) < 1/n^2$ and the
probabilities sum to at most $1$. Summing over all such voters, we get

\[
  \Big|\sum_{i=2}^{n-1} u_i(\de')\Big|
    \;\le\; \frac{n-2}{n^2} \;<\; \frac1n
    \;\xrightarrow{\,n \to \infty\,}\; 0 .
\]
As a result, the total utility in this profile satisfies
\[
  SW(\de') = u_1(\de') + \sum_{i=2}^{n-1} u_i(\de') + u_n(\de') 
    \xrightarrow{\,n \to \infty\,} \lambda\big(1 - e^{-1}\big) \;>\; 0.
\]
So, in particular $SW(\de_{\texttt{SW}}) \ge SW(\de') = \Omega(1)$. Combining
this with the bound on $SW(\de)$ concludes the proof.
\hfill \qed
 
 \subsection*{Proof of \cref{thm:poa upper}}

    We begin by noting that the maximum utility a voter $i$ can get is $\tau_i$. 
This is because, by \cref{cor:optimal-utility-per-voter} (see \cref{app:optimal-paths}), we know how each voter's best delegation path looks like and since distances from a voter are decreasing further along the path, the best case for $i$ is that their vote is being cast with probability 1 at a distance of 0.

Now we will argue that the minimum utility a voter $i$ can get in an equilibrium profile is $\tau_ip_i$. Consider a profile $\de$ that is a Nash equilibrium and towards a contradiction, say that for a voter $i$ it holds $u_i(\de)<\tau_ip_i$. But then $d(i)\neq i$. We now compute the utility that voter $i$ will experience under the profile $(\de_{-i},i)$. This is simply equal to $\tau_ip_i$, contradicting the fact that the considered profile $\de$ is a NE.

As a result, the ratio between the maximum achievable utility and the minimum utility in an equilibrium is at most

\begin{equation*}
    \frac{\sum_{i\in \N}\tau_i}{\sum_{i\in \N}\tau_ip_i} \leq  
        \frac{\sum_{i\in \N}\tau_i}{\sum_{i\in \N}\tau_i \cdot \min_{i\in \N}\{p_i\}}
        =        \frac{1}{\min_{i\in \N}\{p_i\}}.
\end{equation*}



   We now focus on measuring the upper bound of $PoA^+$. Following the same arguments as before, we have that the difference between the welfare in the socially optimal solution and the one in the worst in terms of utility Nash equilibrium is at most 
   
   $$\sum_{i\in\N}\tau_i - \sum_{i\in\N}\tau_ip_i 
   \leq
   \sum_{i\in\N}\tau_i - \sum_{i\in\N}\tau_i\min_{i \in \N}\{p_i\} = (1-\min_{i \in \N}\{p_i\}) \sum_{i\in\N}\tau_i,$$
   which proves the statement.
\hfill $\qed$

\newpage
\section{Additional Concepts and Results}


\subsection{Further Insights On the Structure of Nash Equilibria}

\label{abst:negativestructure}


Consider a mutual acceptance instance $\inst$ with no deterministic voters, a profile $\de$ that is a Nash equilibrium of $\inst$ and a weakly connected component $W$ of $G_{\de}$ that consists of more than a single vertex. Then, the following hold:
\begin{itemize}[leftmargin=*,topsep=0.1cm]  \setlength\itemsep{-0.05cm}
    \item If there is a vertex in $\mathcal{C}(W)$ with $\text{in-deg}=2$, it is not necessarily the left-most or right-most vertex of $\mathcal{C}(W)$. \item Paths appearing in $G_{\de}$ are not necessarily following an order of increasing or decreasing position. 
    \item Vertices in $\mathcal{L}(W)$ or $\mathcal{R}(W)$ do not necessarily form a path in $G_{\de}$.
\end{itemize}

The following instance $\inst,$ where $\N=\{1,2,\cdots,8\}$ proves the first two statements, specifically that the entry points of a cycle are not necessarily the left-most and right-most vertices and that a path ending to a
cycle need not to be between consecutive vertices:

 \begin{align*}
    x&=(0.1, 0.15, 0.2,0.4,0.46,0.63,0.66,0.88),\\
p&=(0.91, 0.15, 0.16,0.33,0.09,0.69,0.47,0.12),\\ \tau_4&=0.5, \tau_i=0.88 \text{ for } i\neq 4.
\end{align*} 
The delegation profile $\de=(5,1,2,7,4,5,6,7)$ is a Nash equilibrium for $\inst$ as \cref{table:prop12tab1} shows. More precisely, in \cref{table:prop12tab1} we can see that no voter has a profitable deviation from  $\de$.

\begin{table*}[h]
\centering
\begin{tabular}{cccccccc}
\toprule
0.8008  & 0.8120  &  0.8214   & 0.8353  &   \textbf{0.8360}  & 0.8313   &  0.8291  & 0.8268
 \\ 
\textbf{0.8073} &  0.1320  &  0.2449   & 0.4959  &  0.5023 &   0.4579   & 0.4377 &   0.4163
 \\ 
0.7734  &  \textbf{0.7831}  & 0.1408   & 0.5382   & 0.5446 &   0.5007  &  0.4807   & 0.4601
 \\ 
0.2893  &  0.2958   & 0.3070 &   0.1650 &   0.1915   & 0.2980  &  \textbf{0.3111}  & 0.2952
 \\ 
0.5098  &  0.5230  &  0.5423   & \textbf{0.6786}  & 0.0792  &  0.5250  &  0.6063  &  0.5933
 \\ 
0.7200  &  0.7216 &   0.7257  &  0.7567   & \textbf{0.7630} &  0.6072  &  0.7310  &  0.7396
 \\ 
0.5797 &   0.5842  &  0.5925  &  0.5220 &   0.5447  &  \textbf{0.7651} & 0.4136  &  0.4556
 \\ 
0.2246  &  0.2265   & 0.2353   & 0.5445  &  0.5378   & 0.5838  &  \textbf{0.6047} &  0.1056
\\ 
\bottomrule
\end{tabular}
\caption{Representation of utilities for potential deviations for the instance proving the first two statements from \cref{abst:negativestructure}. Specifically, cell $(i,j)$ corresponds to the utility of voter $i$ after delegating to voter $j$, assuming that the rest of the delegations are given by $\de_{-i}$. The expected utility for each voter under $\de$ (maximal per row) appears in bold.}
\label{table:prop12tab1}
\end{table*}

The graph $G_{\de}$ is depicted in \cref{fig1} and its structure proves the corresponding statements.

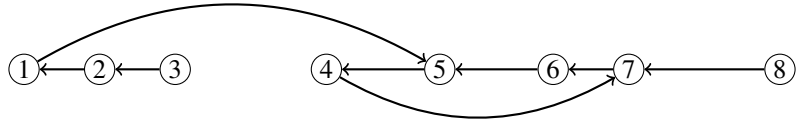
\begin{figure}[h]
\centering
\begin{tikzpicture}[scale=10, every node/.style={circle, draw, minimum size=0.4cm, inner sep=0pt}]
    \node (1) at (0.1, 0) {1};
    \node (2) at (0.2, 0) {2};
    \node (3) at (0.3, 0) {3};
    \node (4) at (0.5, 0) {4};
    \node (5) at (0.65, 0) {5};
    \node (6) at (0.8, 0) {6};
    \node (7) at (0.9, 0) {7};
    \node (8) at (1.1, 0) {8};
    
    \draw[->, thick] (1) to[bend left] (5);
    \draw[->, thick] (2) -- (1);
    \draw[->, thick] (3) -- (2);
    \draw[->, thick] (4) to[bend right] (7);
    \draw[->, thick] (5) -- (4);
    \draw[->, thick] (6) -- (5);
    \draw[->, thick] (7) -- (6);
    \draw[->, thick] (8) -- (7);
\end{tikzpicture}
\caption{Delegation graph of the instance proving the first two statements from \cref{abst:negativestructure}.}
\label{fig1}
\end{figure}

 The third statement, namely that the vertices in
$\mathcal{L}(W)$ or $\mathcal{R}(W)$ do not necessarily form a path in $G_{\de},$ where $\de$ is a NE, holds due to the following instance, where $\N=\{1,2,\cdots,7\}$:

 \begin{align*}
    x&=(0.37,0.4,0.54,0.75,0.87,0.89,0.9),\\
p&=(0.05,0.76,0.03,0.75,0.53,0.93,0.73),\\
\tau&=0.84.
\end{align*}
Consider the profile $\de=(2,4,2,5,7,5,6),$ the delegation graph $G_{\de}$ of which is depicted in \cref{fig2}.

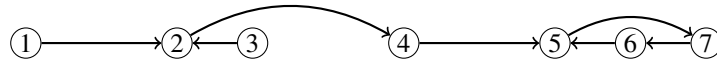
\begin{figure}[h]
\centering
\begin{tikzpicture}[scale=10, every node/.style={circle, draw, minimum size=0.4cm, inner sep=0pt}]
    \node (1) at (0.1, 0) {1};
    \node (2) at (0.3, 0) {2};
    \node (3) at (0.4, 0) {3};
    \node (4) at (0.6, 0) {4};
    \node (5) at (0.8, 0) {5};
    \node (6) at (0.9, 0) {6};
    \node (7) at (1.0, 0) {7};
    
    \draw[->, thick] (1) -- (2); 
    \draw[->, thick] (2) to[bend left] (4); 
    \draw[->, thick] (3) -- (2); 
    \draw[->, thick] (4) -- (5); 
    \draw[->, thick] (5) to[bend left] (7); 
    \draw[->, thick] (6) to (5); 
    \draw[->, thick] (7) to (6); 
\end{tikzpicture}
\caption{Delegation graph of the instance proving the third statement from \cref{abst:negativestructure}.}
\label{fig2}
\end{figure}


\noindent \cref{table:prop12tab2} proves that no voter has a profitable deviation from  $\de$. Notice that vertices in the left of the cycle form a tree.

\begin{table*}[h]
\centering
\begin{tabular}{@{}ccccccc@{}}
\toprule
0.0420 & \textbf{0.7240} & 0.7226 & 0.4468 & 0.3501 & 0.3438 & 0.3366 \\
0.6481 & 0.6384 & 0.6434 & \textbf{0.7478} & 0.7234 & 0.7218 & 0.7200 \\
0.6797 & \textbf{0.6799} & 0.0252 & 0.6030 & 0.5032 & 0.4968 & 0.4894 \\
0.7242 & 0.7231 & 0.7250 & 0.6300 & \textbf{0.8052} & 0.8036 & 0.8017 \\
0.6366 & 0.6383 & 0.6397 & 0.6990 & 0.4452 & 0.8036 & \textbf{0.8199} \\
0.8113 & 0.8117 & 0.8118 & 0.8305 & \textbf{0.8316} & 0.7812 & 0.8236 \\
0.7221 & 0.7235 & 0.7240 & 0.7819 & 0.7291 & \textbf{0.8297} & 0.6132 \\
\bottomrule
\end{tabular}
\caption{Representation of utilities for potential deviations for the instance proving the third statement from \cref{abst:negativestructure}. Specifically, cell $(i,j)$ corresponds to the utility of voter $i$ after delegating to voter $j$, assuming that the rest of the delegations are given by $\de_{-i}$. The expected utility for each voter under $\de$ (maximal per row) appears in bold.}
\label{table:prop12tab2}
\end{table*}

\subsection{Further Insights on the Structure of Optimal Delegation Profiles}\label{sec:structureOpDelegation}




Building on the negative result of \cref{SCno-cycle}, we now identify the condition under which a weakly connected component of $G_{\de_{SW}}$ contains a cycle. 
The next natural question is whether a single cycle passing through all vertices is the optimal solution in such cases, as it is for the delegation profile minimizing vote loss (see \cref{sec:votescast},\cref{sec:votescastexperiments} for details on this metric). We address this to the negative.


\begin{theorem}
\label{thm: sc-structure}
    In a mutual acceptance instance $\inst$ without deterministic voters, if $i\in \mathcal{A}_j(\x,\Btau)$ for every pair $i,j$ of voters of a weakly connected component of a socially optimal delegation profile, then, the component has a cycle. However, it is not the case that a cycle including every vertex corresponds to the delegation profile maximizing social welfare.
\end{theorem}
\begin{proof}
    Towards a contradiction, assume that for an instance $\inst,$ there is a weakly connected component of at least $2$ vertices that doesn't have a directed cycle, therefore, there is a voter $i$ such that $d(i)=i$ in the utility maximizing solution $\de$. Since $i$ belongs to the component, there is another voter, say $j$, such that $d(j)=i$. Consider also the profile $\de'=(\de_{-i},j)$. The delegation graph $G_{\de'}$ contains a cycle in the component of $i$ and $j$. We will prove that, under $\de'$, the total utility of the voters of the component will be strictly greater than under $\de$, while the utility of the rest of the voters remains the same. 

By the fact that $i$ and $j$ both belong to the component, and hence, $j\in \mathcal{A}_i(\x,\Btau)$ it holds that $u_i(\de')=p_i\tau_i+(1-p_i)p_j(\tau_i-\text{dist}(i,j))>p_i\tau_i=u_i(\de).$ Consider now any other arbitrary voter of the examined weakly connected component of $G_{\de}$, say $k \neq i$. It holds that there is a path from $k$ to $i$ in $G_{\de}$, and say that $\mathcal{M}$ corresponds to the set of vertices in that path ($k$ and $i$ included). The change of $d(i)$ from $i$ to $j$ will increase the utility of $k$ by $\prod_{t\in \mathcal{M}}(1-p_t)p_j(\tau_k-\text{dist}(k,j)),$ because with a probability equal to the probability that none in $\mathcal{M}$ will vote and $j$ will vote, voter $k$ will be represented by $j$. By assumptions, this quantity is strictly positive. Therefore, every voter in the component will have a strictly greater utility under $\de'$ compared to their utility under $\de$.

Finally, it is not hard to observe that in the following symmetric instance, the graph $G_{\de},$ where $\de$ is the social welfare maximizing delegation profile, consists of two weakly connected components: One containing the first two and another containing the last two vertices.

 \begin{align*}
    x &= [0.05, 0.06, 0.94, 0.95],\\
p &= [0.5, 0.5, 0.5, 0.5],\\
\tau &= 0.1.
\end{align*}
This is because, due to each voter's tolerance value, a delegation from one of the first two voters to one of the last two, or vice versa, will not appear in the social welfare maximizing profile.
\end{proof}


 \subsection{Individually Optimal Delegation Profiles}
\label{app:optimal-paths}


For an individual voter $i \in \N$, we define $i$'s optimal delegation profile as the profile that maximizes $i$'s expected utility. This profile will be denoted by $\de^{i*}$. It represents the choice $i$ would make, among all possible profiles, if they could fix all voters' delegations to solely maximize their own expected utility. This concept serves as a key metric for evaluating the quality of a delegation profile in terms of social welfare (e.g., whether or not it constitutes a Nash equilibrium).
Interestingly, the cyclic structure also plays a crucial role here. Specifically, we prove that for any voter $i \in \N$, there is no profile that results in strictly higher utility for $i$ than one where the corresponding delegation graph includes a cycle consisting of all voters in $\mathcal{A}_i(\x,\tau)$, ordered by increasing distance to $i$, and closing back to $i$. The utility $i$ obtains from such a profile is equal to the utility they achieve when the edge returning to $i$ is removed.



The main result in this regard is stated below. The individually optimal delegation profile for a voter is referred to as $i$'s optimal delegation path. In principle, specifying this path is sufficient, as delegations involving voters outside this path do not influence $i$'s utility.
Specifically, we show that the most beneficial strategy for a voter is to sequentially delegate their vote to the closest voter, continuing step by step to the farthest voter within a $\tau$ distance.

\begin{theorem}
\label{cor:optimal-utility-per-voter}
    For any $i \in \N$ and instance $\inst$, we have $\mathbb{E}(i,\de) \leq \mathbb{E}(i,\de^{i*})$, where 
    \begin{itemize}
        \item $\de$ is some delegation profile of $\inst$;
        \item $\de^{i*}$ is the profile for which $G_{\de^{i*}}$ includes a path starting at $i$ that traverses all voters in $\mathcal{A}_i(\x,\Btau)$ in order of their distance to $i$.
    \end{itemize}
    That is, $\de^{i*}$ is an optimal delegation profile for voter $i\in\N$.
\end{theorem}

Before proving the result, we revisit our running example to illustrate its statement.
\begin{example1cont}
    By \Cref{cor:optimal-utility-per-voter}, it holds that no better delegation profile exists for voter $D$, than those including either the delegation cycle  $\{C,E,B,D\}$ or $\{E,C,B,D\}$. In both cases, voter $D$ has an optimal expected utility of $0.11196$.
\end{example1cont}

Now we move to proving \cref{cor:optimal-utility-per-voter}. We highlight that by providing a constructive proof for an optimal solution with respect to a single voter, our result contrasts with strategic choice selection in other areas of social choice theory, where the analogous problem is intractable.

\begin{proof}[Proof of \cref{cor:optimal-utility-per-voter}]

    The proof is a direct consequence of a series of three lemmas that follow.
    The first lemma demonstrates that any delegation path of a voter $v$ along which voters are sorted by ascending distance to $v$ yields a higher expected utility than the same delegation path with the order of voters permuted. For a voter $v$ we denote by $\mathbb{E}(v,\pi)$ the (expected) utility that $v$ gains from a delegation profile the graph of which has $\pi$ as the maximal path from $v$.

\begin{lemma}\label{le:order-del-by-dist}
    Let $\pi = (v_1, \cdots, v_k)$ be a delegation path for a voter $v_1$. Assume there exists an index $i$ such that $v_{i+1}$ is closer to $v_1$ than $v_i$, i.e., $\text{dist}(v_1,v_{i+1}) < \text{dist}(v_1,v_{i})$. Then the expected utility $v_1$ gains from $\pi$ is lower than the expected utility from path $\pi' = (v_1, \cdots, v_{i-1}, v_{i+1}, v_i, v_{i+2}, \cdots, v_k)$ where the delegation path first visits $v_{i+1}$ and then $v_i$.
\end{lemma}
\begin{proof}
    We compare the expected utility of $v_1$ with respect to  $\pi$
    
    \[\mathbb{E}(v_1,\pi) = \sum\limits_{\ell=1}^{k}(\tau_1- \text{dist}(v_1, v_\ell))p_{v_\ell}\prod\limits_{m=1}^{\ell-1}(1-p_{v_m})\]
    
\noindent    to the expected utility of $v_1$ with respect to  $\pi'$
    
     \begin{gather*}
        \mathbb{E}(v_1,\pi') =  \quad [\sum\limits_{\ell=1}^{i-1}(\tau_1- \text{dist}(v_1, v_\ell))p_{v_\ell}\prod\limits_{m=1}^{\ell-1}(1-p_{v_m})] + (\tau_1- \text{dist}(v_1, v_{i+1}))p_{v_{i+1}}\prod\limits_{m=1}^{i-1}(1-p_{v_m}) \\
         + (\tau_1- \text{dist}(v_1, v_{i}))p_{v_{i}}(1-p_{v_{i+1}})\prod\limits_{m=1}^{i-1}(1-p_{v_m})  + \sum\limits_{\ell=i+2}^{k}(\tau_1- \text{dist}(v_1, v_\ell))p_{v_\ell}\prod\limits_{m=1}^{\ell-1}(1-p_{v_m}).
    \end{gather*}
    
    These expected utilities only differ in terms that include the utilities of $i$ and $v_{i+1}$. Further, these terms have common factors. 
    
     \begin{gather*}         \mathbb{E}(v_1,\pi)<\mathbb{E}(v_1,\pi') 
        \Leftrightarrow  \\
        (\tau_1- \text{dist}(v_1, v_{i}))p_{v_{i}}+ (\tau_1- \text{dist}(v_1, v_{i+1}))p_{v_{i+1}}(1-p_{v_{i}}) 
         \\ <(\tau_1- \text{dist}(v_1, v_{i+1}))p_{v_{i+1}}+ (\tau_1- \text{dist}(v_1, v_{i}))p_{v_{i}}(1-p_{v_{i+1}})
        \Leftrightarrow  \\
          \tau_1 p_{v_{i}} - \text{dist}(v_1, v_{i})p_{v_{i}} + \tau_1 p_{v_{i+1}} - \tau_1 p_{v_{i+1}}p_{v_{i}} - \text{dist}(v_1, v_{i+1}) p_{v_{i+1}} + \text{dist}(v_1, v_{i+1}) p_{v_{i+1}} p_{v_{i}} <\\
           \tau_1 p_{v_{i+1}} - \text{dist}(v_1, v_{i+1})p_{v_{i+1}} + \tau_1 p_{v_{i}} - \tau_1 p_{v_{i}}p_{v_{i+1}} - \text{dist}(v_1, v_{i}) p_{v_{i}} + \text{dist}(v_1, v_{i}) p_{v_{i}} p_{v_{i+1}}
        \Leftrightarrow  \\
           \text{dist}(v_1, v_{i+1}) p_{v_{i+1}} p_{v_{i}} <  \text{dist}(v_1, v_{i}) p_{v_{i}} p_{v_{i+1}} 
        \Leftrightarrow          \text{dist}(v_1, v_{i+1}) <  \text{dist}(v_1, v_{i})  \qedhere
     \end{gather*}   
\end{proof}

Note that \Cref{le:order-del-by-dist} holds even for paths that may include voters outside their $\tau$-threshold, i.e., contribute negative utility.
However, as we see next, eliminating such voters from (the end of) a delegation path increases the expected utility.

\begin{lemma}\label{le:only-del-within-alpha-dist}
    Let $\pi = (v_1, \cdots, v_k)$ be a delegation path for a voter $v_1$ and let $\pi' = (v_1, \cdots, v_{k-1})$ be the shortened path by deleting the last voter $v_{k}$. Then the expected utility of $v_1$ with respect to  delegation path $\pi$ is lower than that for path $\pi'$, i.e., $\mathbb{E}(v_1,\pi) < \mathbb{E}(v_1,\pi')$, if and only if $\tau_1- \text{dist}(v_1, v_{k}) < 0$.
\end{lemma}
\begin{proof}
    We compare the expected utility of $v_1$ with respect to  $\pi$ and $\pi'$: 
    
     \begin{align*}
         \mathbb{E}(v_1,\pi)  &<\mathbb{E}(v_1,\pi')  \Leftrightarrow \\
         \sum\limits_{\ell=1}^{k}(\tau_1- \text{dist}(v_1, v_\ell))p_{v_\ell}\prod\limits_{m=1}^{\ell-1}(1-p_{v_m}) 
         &< \sum\limits_{\ell=1}^{k-1}(\tau_1- \text{dist}(v_1, v_\ell))p_{v_\ell}\prod\limits_{m=1}^{\ell-1}(1-p_{v_m}) \Leftrightarrow \\
           (\tau_1- \text{dist}(v_1, v_{k}))p_{v_{k}}\prod\limits_{m=1}^{k-1}(1-p_{v_m})  &<  0 \Leftrightarrow
           \tau_1- \text{dist}(v_1, v_{k})  <  0,
    \end{align*}
    hence the result follows.
\end{proof}

Thus, to increase the expected utility, we can first sort voters along a delegation path in order of their distance and then repeatedly shorten the path to only contain voters with a distance smaller or equal $\tau$. 
We now show that inserting a voter into a path increases the expected utility if the subsequent voters have a larger distance.

\begin{lemma}\label{le:include-all-in-alpha-range}
    Let $\pi = (v_1, \cdots, v_k)$ be a delegation path for a voter $v_1$. Assume there exists an index $i$ such that $\text{dist}(v_1,v_{i}) < \text{dist}(v_1,v_{\ell})$ for all $\ell \in[i+1,k]$.
    Let $\pi' = (v_1, \cdots, v_{i-1}, v_{i+1}, \cdots, v_k)$ where $i$ is deleted. Then the expected utility of $v_1$ with respect to  delegation path $\pi$ is higher than that for path $\pi'$, i.e., $\mathbb{E}(v_1,\pi) > \mathbb{E}(v_1,\pi')$. 
\end{lemma}
\begin{proof}
    We compare the expected utility of $v_1$ with respect to  $\pi$
    
    \[\mathbb{E}(v_1,\pi) = \sum\limits_{\ell=1}^{k}(\tau_1- \text{dist}(v_1, v_\ell))p_{v_\ell}\prod\limits_{m=1}^{\ell-1}(1-p_{v_m})\]
    
\noindent    to the expected utility of $v_1$ with respect to  $\pi'$
    
     \begin{gather*}
        \mathbb{E}(v_1,\pi') =  
\sum\limits_{\ell=1}^{i-1}(\tau_1- \text{dist}(v_1, v_\ell))p_{v_\ell}\prod\limits_{m=1}^{\ell-1}(1-p_{v_m})  + \sum\limits_{\ell=i+1}^{k}(\tau_1- \text{dist}(v_1, v_\ell))p_{v_\ell}\prod\limits_{m=1}^{i-1}(1-p_{v_m})\prod\limits_{m=i+1}^{\ell-1}(1-p_{v_m}).
    \end{gather*}
    These expected utilities only differ in the term that includes utilities of $i$ and some terms that have an additional factor $(1-p_{i})$. Thus, we have: 
    
     \begin{align*}
         \mathbb{E}(v_1,\pi) & >\mathbb{E}(v_1,\pi') \Leftrightarrow \\
          \sum\limits_{\ell=i}^{k}(\tau_1{-} \text{dist}(v_1, v_\ell))p_{v_\ell}\prod\limits_{m=1}^{\ell-1}(1{-}p_{v_m}) 
        & > \sum\limits_{\ell=i+1}^{k}(\tau_1- \text{dist}(v_1, v_\ell))p_{v_\ell}\prod\limits_{m=1}^{i-1}(1-p_{v_m})\prod\limits_{m=i+1}^{\ell-1}(1-p_{v_m}) \Leftrightarrow \\
        \sum\limits_{\ell=i}^{k}(\tau_1{-} \text{dist}(v_1, v_\ell))p_{v_\ell}\prod\limits_{m=i}^{\ell-1}(1{-}p_{v_m}) 
        & > \sum\limits_{\ell=i+1}^{k}(\tau_1- \text{dist}(v_1, v_\ell))p_{v_\ell}\prod\limits_{m=i+1}^{\ell-1}(1-p_{v_m}) \Leftrightarrow\\
          (\tau_1- \text{dist}(v_1, i))p_{i} + (1-p_{i})
        & > \sum\limits_{\ell=i+1}^{k}(\tau_1- \text{dist}(v_1, v_\ell))p_{v_\ell}\prod\limits_{m=i+1}^{\ell-1}(1-p_{v_m}) \Leftrightarrow\\
         (\tau_1- \text{dist}(v_1, i))p_{i}
        & > (1-(1-p_{i}))\sum\limits_{\ell=i+1}^{k}(\tau_1{- }\text{dist}(v_1, v_\ell))p_{v_\ell}\prod\limits_{m=i+1}^{\ell-1}(1{-}p_{v_m}) \Leftrightarrow\\
      \tau_1- \text{dist}(v_1, i)
        & > \sum\limits_{\ell=i+1}^{k}(\tau_1- \text{dist}(v_1, v_\ell))p_{v_\ell}\prod\limits_{m=i+1}^{\ell-1}(1-p_{v_m})
    \end{align*}
    To see that the last inequality holds, we first show the following:
    
     \begin{align*}
        & \sum\limits_{\ell=i+1}^{k}p_{v_\ell}\prod\limits_{m=i+1}^{\ell-1}(1-p_{v_m})= p_{v_{i+1}} + (1-p_{v_{i+1}})\sum\limits_{\ell=i+2}^{k}p_{v_\ell}\prod\limits_{m=i+2}^{\ell-1}(1-p_{v_m})=\\
        & p_{v_{i+1}} + (1-p_{v_{i+1}})(p_{v_{i+2}} + (1-p_{v_{i+2}})\sum\limits_{\ell=i+3}^{k}p_{v_\ell}\prod\limits_{m=i+3}^{\ell-1}(1-p_{v_m})) =       \\
        &\cdots\\
        & p_{v_{i+1}} + (1-p_{v_{i+1}})(p_{v_{i+2}} + (1-p_{v_{i+2}})(\cdots (p_{v_{k-1}} + (1-p_{v_{k-1}})p_{v_k}))) \leq        \\
        & p_{v_{i+1}} + (1-p_{v_{i+1}})(p_{v_{i+2}} + (1-p_{v_{i+2}})(\cdots (p_{v_{k-1}} + (1-p_{v_{k-1}})\cdot 1)))        \\
     \end{align*}

\noindent         The previous bound was obtained by bounding the last factor: $p_{v_k}\leq 1$. This leaves this multiplicative term to be: $p_{v_{k-1}} + (1-p_{v_{k-1}})=1$. Observe that this step will now repeat throughout the terms for decreased subscripts. We see the final steps of this recursive process in the following lines. 
         
     \begin{align*}
        \cdots
        = p_{v_{i+1}} + (1-p_{v_{i+1}})(p_{v_{i+2}} + (1-p_{v_{i+2}})\cdot 1) =                p_{v_{i+1}} + (1-p_{v_{i+1}})\cdot 1    =1    
    \end{align*}
    
\noindent    Together with our assumption that $\text{dist}(v_1,v_{i}) < \text{dist}(v_1,v_{\ell})$ for all $i<\ell$, this yields the desired inequality:
    
     \begin{align*}
         \sum\limits_{\ell=i+1}^{k}(\tau_1- \text{dist}(v_1, v_\ell))p_{v_\ell}\prod\limits_{m=i+1}^{\ell-1}(1-p_{v_m}) 
        < &  \sum\limits_{\ell=i+1}^{k}(\tau_1- \text{dist}(v_1, i))p_{v_\ell}\prod\limits_{m=i+1}^{\ell-1}(1-p_{v_m}) =\\
          (\tau_1- \text{dist}(v_1, i))\sum\limits_{\ell=i+1}^{k}p_{v_\ell}\prod\limits_{m=i+1}^{\ell-1}(1-p_{v_m}) 
        \leq &  \tau_1- \text{dist}(v_1, i).\qedhere
    \end{align*}
\end{proof}

Combining the three lemmas above, one can show that for an individual voter $v$, the maximal expected utility can be achieved by a delegation path that traverses all voters within $\tau_v$-distance and in order by their distance to $v$.
\end{proof}

From \cref{cor:optimal-utility-per-voter} we can directly get the following result, providing an upper bound for the optimal social welfare.
\begin{corollary}\label{lem:best-path-higher-EU}
    For any instance $\inst=\langle\x, \p, \Btau\rangle$, if $\de_{SW}= \argmax_{\de \in \N \times \cdots \times \N} SW(\de)$ is an optimal profile and $\de^{i*}$ is the delegation profile with voter $i$'s highest expected utility, then $SW(\de_{SW}) \leq \sum_{i \in \N} u_i(\de^{i*})$.
\end{corollary}

The result of \Cref{lem:best-path-higher-EU} provides an upper bound on the highest possible social welfare for any delegation profile in an instance. 
Let $\ODP(\inst)$ denote the sum of expected utilities across the \emph{optimal delegation profiles} of all voters in instance $\inst$, i.e., $\ODP(\inst) = \sum_{i \in \N} u_i(\bf{d^{i*}})$. This approximation is employed in experiments involving socially optimal profiles, where finding the exact solution is computationally expensive (see \Cref{sec:PoA-exp}).
One natural question is how far away is $\ODP(\inst)$ from $SW(\de_{SW})$. We conducted some simulations on small random instances where the delegation profile with the highest social welfare can be effectively found by brute force. Our preliminary experiments showed that $\ODP(\inst)$ is just a few percentiles larger than its corresponding  $SW(\de_{SW})$. Details can be found in \Cref{app:ApproximatingSWex}.

\subsection{Expected Number of Votes Cast}\label{sec:votescast}

The concept of lost votes is crucial in liquid democracy research, as it has been a key motivation behind the scheme since its inception. In our work, among others, we are evaluating delegation profiles, both theoretically and experimentally, based on their potential to mitigate the loss of voting power over all elections held in the system. 
By \emph{votes lost}, we refer to the expected number of votes not cast due to voters abstaining from voting in a given election, which in turn results in the loss of the voting power of those who delegated to them as well.
For simplicity, consider a profile $\de$ where each weakly connected component of $G_{\de}$ contains a cycle. Fix such a delegation cycle $\mathcal{C} = (v_1, \cdots, v_k)$. Since the out-degree of every vertex in $G_{\de}$ is 1, $\mathcal{C}$ has no outgoing edges. 
Consequently, no votes from voters corresponding to vertices in the component of $\mathcal{C}$ are lost as long as at least one voter in $\mathcal{C}$ casts a ballot. Equivalently, votes within the connected component can only be lost if all voters in $\mathcal{C}$ abstain. The expected number of votes lost thus depends on the probability that the delegation cycle remains unbroken and the expected number of incoming votes. A similar analysis applies to components that terminate in self-loops. It holds that

 \begin{align*}
\mathds{E}[\# \text{votes lost}] = \Pi_{i = 1, \cdots, k} (1-p_{i}) \cdot  
 \left(k{+}\sum_{i=1,\cdots,k}\mathds{E}[\# \text{votes (from outside C) delegated to } v_i]\right).
\end{align*}
Note that here, the term $\mathds{E}[\# \text{ votes delegated to } i]$ is a computation over a tree in $G_{\de}$. 
That is, we can compute the expected number of votes delegated to some voter recursively via the expected number of votes delegated to their direct predecessors. 

 \begin{align*}
\mathds{E}[\# \text{ votes delegated to } v] = 
\sum_{w \in V: d(w)=v} (1-p_w) \cdot 
 \left(1 + \mathds{E}[\# \text{ votes delegated to } w]\right)
\end{align*}
This is computationally feasible via a recursion starting from voters with no in-delegations.
Observe that, finally, the expected number of votes cast is $n-\mathds{E}[\# \text{ votes lost}]$. 
Given the previous definitions, the following structural result regarding the minimization of lost voting power is straightforward.

\begin{proposition}
The delegation graph of the profile that minimizes the number of lost votes consists of a single cycle passing through every vertex in arbitrary order.
\end{proposition}

\newpage

\section{Experimental Analyses of Our Model}\label{app:experi}

In this section of the Appendix, we present the various experimental analyses conducted to support the theoretical results of our model. Note that there are files in the supplementary material corresponding to each of the following sections

 \subsection{Best-Response Dynamics}\label{app:BRD}

With a definition of Nash equilibria in our model, we now need a procedure to find them when they exist. We naturally consider a best-response dynamic, which works by repeatedly checking if any voter has a better delegation choice based on the current profile of delegations. A detailed pseudocode of the best-response dynamic (BR) is being presented as \cref{Alg:BRlongterm}. The process starts with an arbitrary profile of delegations, and then updates the delegations of voters one by one. Each voter updates their delegation to their best response when one exists, with ties broken arbitrarily. We will refer to each instance where the protocol checks if a voter has a best response as a \emph{round}. The algorithm will stop only after $n$ rounds without any voter finding a best response. This ensures that the resulting delegations form a Nash equilibrium, as no voter can improve their expected utility by unilaterally changing their delegation.

\begin{algorithm}[h!]
\caption{BR Protocol}\label{Alg:BRlongterm}
\begin{algorithmic}[1]
 \State Input:  $\p$, $\x$, $\Btau$ and $n$
  \State\emph{Initialize a random delegation vector $\de$}
  \State $it=0$ 
  \While{$it \leq n$}
    \State{$\de^1= \de$} 
      \For{each $i\in \N$}\label{alg:BRFor}
      \State $it=it+1$ 
        \For{each $j\in \N$}
            \State $\de'=(\de_{-i}^1,j)$
            \State $exp(j)= u_i(\de')$
        \EndFor
        \If {$\max_{j\in \N} (exp(j)) > u_i(\de^1)$}
        \State $\de(i)= \arg\max_{j\in \N}(exp(j))$
        \State $it=0$\label{line:it0} 
        \ElsIf{$it \geq n$} 
        \State  \Return{$\de$}
        \EndIf
     \EndFor
 \EndWhile
\end{algorithmic}
\end{algorithm}

\paragraph{Experimental Analysis of Symmetric Instances.}We ran the BR protocol on 20,000 instances with symmetric $\Btau$ such that the parameters were chosen uniformly at random as follows:
\begin{itemize}
    \item $n$ chosen randomly from $\{1,\cdots, 100\}$ 
    \item $\tau\in [0,\frac{2}{3}]$ rounded to two decimals.
    \item $\x\in[0,1]^n$ rounded to two decimals.
    \item $\p\in[0,1]^n$ rounded to two decimals.
\end{itemize}
We then ran the BR protocol starting from a random initial profile of delegations. For each of these instances, the BR protocol found a NE. 

\begin{example}\label{ex:BRprotocol}
Consider six voters $\N=\{A,B,C,D,E,F\}$ whose opinions can be placed on a line such that $\x=(0.2,0.25, 0.4, 0.4, 0.6, 0.8)$, probabilities $\p=(0.5, 0.5, 0.9, 0.3, 0.5, 0.3)$, and tolerance $\Btau=(0.25, 0.25, 0.25, 0.25, 0.25, 0.25)$. 

We will apply the best response protocol to our running example, first presented in \Cref{Longterm:ex:start}. We assume the randomly initiated delegations are $\de=(A,D,E,B,F,E)$. Following the BR protocol, we start with voter $A$ and check if there is a BR. $A$'s BR   is delegating to voter $B$. We then check for voter $B$, who has a BR to delegate to $A$. We then inspect $C$ and see that their BR is to delegate to $D$. This continues until we arrive at a NE $\de''=(B,A,D,C,C,F)$ after $6$ iterations. In these first $6$~steps, the counter $it$ is set to $0$ each time a voter updates their delegations with a BR. The protocol will continue for another $6$~iterations with no more updates (as we have reached a NE). The protocol terminates when $it=6$. Moreover, we note that for our example, there are only two NE, $\de'$ and $\de''$. 
\end{example}


\subsection{Creating General Instances of our Model}\label{app:setup}


\paragraph{Varying the number of voters.}
Our first set of instances varies the number of voters $n \in \{20, 50, 100, 200\}$. 
For each value of $n$, we created {$100$} instances, randomly selecting $\x\in [0,1]^n$ to three decimal places, $\p\in [0,1]^n$ rounded to two decimal places,  and $\Btau\in[0,1]^n$ rounded to one decimal place. 
Then for each instance $\inst=(\x,\p,\Btau)$, we find a delegation profile $\dne$ which is a Nash equilibrium that is found via the best response protocol (given as \cref{Alg:BRlongterm}, in \Cref{app:BRD}). 

\paragraph{Varying the size of $\Btau$.}
The second set of instances evaluates the impact of the size of tolerance vectors on the various measures examined in the remainder of the experiments.   
We take the $100$ previous instances when $n=50$ and then create $5$ tolerance vectors $\Btau\in [0,1]^n$; we will denote these vectors as $\Btau_{\leq 1}$ (no restriction on the number of decimal places). 
We then modify each of the $5$ vectors $\Btau_{\leq 1}$ by scaling each $\Btau_{\leq 1}(i)$ by $0.75$ and $0.5$, i.e., $\Btau_{\leq 0.75}(i)=0.75\times\Btau_{\leq 1}(i)$ and $\Btau_{\leq 0.5}(i)=0.5\times\Btau_{\leq 1}(i)$. 
Resulting in $5$ vectors $\Btau_{\leq 0.75}$ and $5$ vectors $\Btau_{\leq 0.5}$ for each of the $100$ instances when $n=50$. 
Therefore, $500$ instances of $\inst=(\x, \p, \Btau_{\leq 1})$, $\inst=(\x, \p, \Btau_{\leq 0.75})$, and $\inst=(\x, \p, \Btau_{\leq 0.5})$. 
For each of these instances, we again employ the best response protocol to determine a delegation profile that constitutes a Nash equilibrium. 

\paragraph{Delegation profiles for different voting models.}
When analyzing the expected number of votes lost (\Cref{sec:votescastexperiments}) and the proportion of SW achieved (\Cref{sec:quality}), we wanted to compare an arbitrary NE $\dne$ with delegation profiles that reflect different voting models. 
The first is acyclic liquid democracy \texttt{acyc}, which replicates the models of liquid democracy that do not permit cycles. We create a delegation profile $\de_{\texttt{acyc}}$ for each of the instances mentioned previously in the section that modifies the corresponding $\dne$ by breaking each cycle at a random point and replacing the delegation within the cycle with a self-loop. 
The second delegation profile models the direct democracy voting model without delegations. We model this with the delegation profile $\de_{\texttt{dir}}$ where every voter delegates to themselves.

\paragraph{Technical Specifications}
The simulations were run on a PC with $32$ GB of RAM, an Intel(R) Core(TM) i7-14700 processor, and a $128$ MB Intel(R) UHD Graphics 770 card.

We used the following library in our simulations (their versions are included in brackets): numpy (1.26.4), tqdm (4.66.4), matplotlib (3.8.4), pandas (2.2.2), networkx (3.2.1), sys (3.12.4, packaged by Anaconda).

\subsection{Experimental Analysis of the Structure of NE}
\label{sec:exp-structure}


\begin{table}[t]
    \centering
    \begin{tabular}{lcccc}
    \toprule
    Number of voters  & 20 & 50   &  100 &200   \\
    \midrule
      Number of SCCs 
    &  3.83 & 6.4& 11.5& 20.93\\
      Number of cycles 
      &2.5& 3.97& 6.47& 10.97\\
      Number of self loops
      &1.33& 2.43& 5.03& 9.97\\
      Number of WCCs of size $1$   
      &1.2& 2.17& 4.4& 9.7\\
      Number of paths into SCC&
      3.03& 6.5& 11.4& 19.37\\
     Average width of cycle & 
     0.09& 0.04& 0.03& 0.02\\
     Average width of WCC &0.29& 0.19& 0.12& 0.08\\
     Average size of cycle & 3.89& 4.68& 5.14& 5.73\\
     \bottomrule
    \end{tabular}
    \caption{ Reporting various average measurements of the structure of its Nash equilibria taken over the {$100$} instances of each size (i.e., $\N\in \{20, 50, 100, 200\}$). All values are rounded to 2 decimal places. SCC denotes strongly connected components (i.e., cycles and self-loops), and WCC denotes weakly connected components (i.e., an SCC and the paths entering the SCC).  }
    \label{tab:standardStructure}
\end{table}

We complement our theoretical results with simulations illustrating how delegation graphs of Nash equilibria, found using our best-response dynamic (\Cref{app:BRD}), can appear in synthetic instances. 
We create instances at random of various sizes: $n \in \{20, 50, 100, 200\}$.
For each instance size, we create {$100$} instances comprised of $\x, \p$ and $\Btau$ chosen uniformly at random such that for each $i\in  \N$,  $p_i, \tau_i\in [0,1]$ (rounded to 2 d.p.),  and $x_i\in [0,1]$ (rounded to 3 d.p.). 
In total, we have {$400$} instances for which we ran our best-response protocol and found a NE, analyzing its structure according to average measurements for various metrics (see \cref{tab:standardStructure}). 
First, observe that all the values in the table represent the average values over the $100$ instances with the same number of voters. For example, on average, each of the 100 instances when $n=20$ have $3.83$ cycles, and when $n=50$, the average width of a cycle, averaged over all {$100$} instances, is $0.04$.  
We observe that the number of voters in a component grows with $n$, as well as the number of cycles. 
Both the width of cycles and components (i.e., maximum distance between voters therein) decrease as the number of voters increases, with voters in the same component, especially in cycles, having very close positions. The proportion of voters who prefer self-delegation over participating in a cycle remains steady at approximately $5\%$, similar to the proportion in weakly connected components of~size~$1$.

\begin{table}
    \centering
    \begin{tabular}{lccc}
    \toprule
         &$\Btau_{\leq 1}$  & $\Btau_{\leq 0.75}$  & $\Btau_{\leq 0.5}$ \\
         \midrule
      Average tolerance & 0.50  &0.37  & 0.25 \\
     Number of SCCs & 5.55& 6.32 &7.60\\
     Number of cycles & 
     4.59 &5.07& 5.78\\
     Number of self loops  & 
     0.96 &1.25 &1.82\\
     Number of WCCs of size $1$ & 
     0.86 &1.10& 1.54\\
     Average number of paths into SCC& 
     6.73 &7.15 &7.79\\
     Average width of SCC   & 
     0.59 &0.57 &0.53\\
     Average width of WCC   & 
     0.77&0.74 &0.70\\
     Average size of cycle & 
     4.54 &4.25 &3.91\\
     \bottomrule
    \end{tabular}
    \caption{Reporting various average measurements of the structure each of the $500$ equilibria found for each tolerance vector $\Btau_{\leq k }$ with $k\in \{0.5, 0.75, 1\}$. Note that SCC denotes strongly connected components (i.e., cycles and self-loops), and WCC denotes weakly connected components (i.e., an SCC and the paths entering the SCC).
    }
    \label{tab:taustacture}
\end{table}
The second way we analyze the structure of a NE within the default delegation model is via the impact of the randomly chosen values of $\Btau$. Thus, we restrict the following experiments to our previous {$100$} instances when $n=50$, and we then study the impact of scaling the tolerance vectors by $0.75$ and $0.5$. 
We take $5$ vectors $\Btau_{\leq 1}\in[0,1]^n$ chosen uniformly at random for each pair $\x, \p$.  We let $\Btau_{\leq 0.75}= 0.75\times \Btau$ and $\Btau_{\leq 0.5}= 0.5\times \Btau$, scaling each value in the vector by either $0.75$ or $0.5$.
For $\x, \p$ and each $\Btau_{\leq k}$ with $k\in \{1,0.75, 0.5\}$, a NE is found via our best response protocol (see \Cref{Alg:BRlongterm} in \Cref{app:BRD}). We then take the average measurements over the 500 equilibria found from a certain type of $\Btau$. These values can be found in \Cref{tab:taustacture}. 



Observe that the number of connected components increases while the values in the tolerance measures decrease, and we see that the number of connected components ending in a self-loop increases slightly as well. Moreover, the number of connected components of size one, i.e., voters delegating to themselves without receiving any delegations, also increases as the tolerance vector decreases.
Another measure we examined was the width of a cycle and a weakly connected component, i.e., the largest distance between any two voters within it. The width of cycles and weakly connected components and the size of cycles all decrease slowly as the tolerance vectors decrease, consistently maintaining notably small widths. 


\subsection{Experimental Analyses of Approximating SW}\label{app:ApproximatingSWex}

Our experimental setup created 400 instances for each \( n \in \{5, 6, 7, 8\} \), where each \( i \in \N\) had a distinct position \( x_i \in [0, 1] \) rounded to three decimal places where no two voters are at the same position, a randomly chosen \( \tau_i \in (0, 1) \) without rounding, and a \( p_i \in (0, 1) \) rounded to two decimal places. 
For each of these instances, we identified the delegation profile \( \de_{SW} \) that maximized \( SW(\de_{SW}) \) and computed \( \ODP(\inst) \). We then examined the values of \( \nicefrac{SW(\de_{SW})}{n} \) and \( \nicefrac{\ODP(\inst)}{n} \), which allowed us to calculate the increase in the average expected utility of a voter in \( \de_{SW} \) compared to their optimal path. This led to \( \nicefrac{\ODP(\inst)}{n} \) being 3.0\%, 3.0\%, 2.8\%, and 2.6\% higher than \( \nicefrac{SW(\de_{SW})}{n} \) for \( n = 5, 6, 7, 8 \), respectively. Thus, our upper bound \( \ODP(\inst) \) is close to the highest possible social welfare \( SW(\de_{SW}) \). However, there appears to be a trend of decreasing distance as the number of voters grows.


\begin{table}[t!]
    \centering
\begin{tabular}{cccc}
\toprule
     & $\dne$ &$\de_{\texttt{acyc}}$&$\de_{\texttt{dir}}$\\
     \midrule
   20 voters  & 0.95&0.86&0.50\\
   50 voters  &0.97 &0.90&0.49\\
   100 voters  &0.97 &0.93&0.50\\
   200 voters  &0.97&0.93&0.50\\
   
    $\Btau_{\leq 1 }$ &0.90&0.80&0.38\\
    $\Btau_{\leq 0.75 }$ &0.88 &0.78&0.38\\
   $\Btau_{\leq 0.5 }$  &0.85&0.75&0.38\\
   \bottomrule
\end{tabular}
    \caption{The average percentage of the votes cast across the {$100$} instances when varying $n\in \{20,50, 100, 200\}$ and $500$ instances when varying $\Btau_{\leq k }$ for $k \in \{1, 0.75, 0.5\}$. The three models we compare the percentage of votes cast are the equilibria within our default delegation model (\texttt{NE}), acyclic liquid democracy (\texttt{acyc}), and direct democracy (\texttt{dir}).}
    \label{tab:Votecast}
\end{table}

\subsection{Expected Number of Votes Cast in Equilibria}\label{sec:votescastexperiments}

    
 
We now revisit our instances with varying $n$ and $\Btau_{\leq k }$ (from \cref{app:setup}), which include $100$ and $500$ instances for each variant, respectively. To contextualize our model, we compare the expected percentage of votes cast under three scenarios: a NE default delegation model $\texttt{NE}$, acyclic liquid democracy $\texttt{acyc}$, and direct democracy \texttt{dir}. We use the delegation profiles for the three models described in \Cref{app:setup}. 
We computed the average expected percentage of votes cast in each instance for all three models, and the results are presented in \Cref{tab:Votecast}. 

When varying the number of voters, we observe that in the default delegation model, a random NE achieves a very high percentage of votes being cast, starting at $95\%$ for 20 voters and steadily increasing as the number of voters grows. A similar pattern is seen in the \texttt{acyc} setting, though the increase is sharper, and, even at its peak (with 200 voters) the percentage of votes cast in the \texttt{acyc} setting remains below the lowest value achieved in the default delegation model (achieved for 20 voters). In stark contrast, the direct democracy setting consistently loses around half of the votes, which corresponds to the average probability of abstaining. These results highlight the significant advantage of our framework in preserving voting power. 
When varying the values of $\Btau_{\leq k }$ for $k \in \{1,0.75, 0.5\}$ leads to $90\%$, $88\%$, and $85\%$ expected votes cast in the default delegation model.


\end{document}

\end{document}